\documentclass{egpubl}
\usepackage{eg2020}
 
\usepackage{amsmath,amsfonts}
\usepackage{wrapfig}
\usepackage{subcaption}
\usepackage{float}
\usepackage{textcomp}
\usepackage{import}
\usepackage{soul}
\usepackage{graphics}

\DeclareMathOperator{\grad}{grad}
\DeclareMathOperator{\re}{Re}
\DeclareMathOperator{\im}{Im}
\newcommand{\toremove}[1]{{}}

\newtheorem{theorem}{Theorem}[section]

 \JournalPaper         %
\usepackage[T1]{fontenc}
\usepackage{dfadobe}  

\usepackage{cite}  %
\BibtexOrBiblatex%
\electronicVersion
\PrintedOrElectronic
\ifpdf \usepackage[pdftex]{graphicx} \pdfcompresslevel=9
\else \usepackage[dvips]{graphicx} \fi

\usepackage{egweblnk} 
\title[A curvature and density-based generative representation of shapes]%
      {A curvature and density-based generative representation of shapes}

\author[Z. Ye \& N. Umetani \& T. Igarashi \& T. Hoffmann]
{\parbox{\textwidth}{\centering Z. Ye$^{1}$
        and N. Umetani$^{2}$
and T. Igarashi$^{2}$
and T. Hoffmann$^{1}$
        }
        \\
{\parbox{\textwidth}{\centering $^1$TU Munich, Germany\\
         $^2$ The University of Tokyo, Japan}
}
}

\begin{document}

\teaser{
\centering
    \def\svgwidth{0.95\textwidth}
\begingroup%
  \makeatletter%
  \providecommand\color[2][]{%
    \errmessage{(Inkscape) Color is used for the text in Inkscape, but the package 'color.sty' is not loaded}%
    \renewcommand\color[2][]{}%
  }%
  \providecommand\transparent[1]{%
    \errmessage{(Inkscape) Transparency is used (non-zero) for the text in Inkscape, but the package 'transparent.sty' is not loaded}%
    \renewcommand\transparent[1]{}%
  }%
  \providecommand\rotatebox[2]{#2}%
  \newcommand*\fsize{\dimexpr\f@size pt\relax}%
  \newcommand*\lineheight[1]{\fontsize{\fsize}{#1\fsize}\selectfont}%
  \ifx\svgwidth\undefined%
    \setlength{\unitlength}{732.95894764bp}%
    \ifx\svgscale\undefined%
      \relax%
    \else%
      \setlength{\unitlength}{\unitlength * \real{\svgscale}}%
    \fi%
  \else%
    \setlength{\unitlength}{\svgwidth}%
  \fi%
  \global\let\svgwidth\undefined%
  \global\let\svgscale\undefined%
  \makeatother%
  \begin{picture}(1,0.34468595)%
    \lineheight{1}%
    \setlength\tabcolsep{0pt}%
    \put(0.3800692,0.32481861){\color[rgb]{0,0,0}\makebox(0,0)[lt]{\lineheight{1.25}\smash{\begin{tabular}[t]{l}OGN\end{tabular}}}}%
    \put(0.00482901,0.33432022){\color[rgb]{0,0,0}\makebox(0,0)[lt]{\lineheight{1.25}\smash{\begin{tabular}[t]{l}Ground Truth\end{tabular}}}}%
    \put(0.5154309,0.32160269){\color[rgb]{0,0,0}\makebox(0,0)[lt]{\lineheight{1.25}\smash{\begin{tabular}[t]{l}point-cloud AE\end{tabular}}}}%
    \put(0.19588432,0.32452626){\color[rgb]{0,0,0}\makebox(0,0)[lt]{\lineheight{1.25}\smash{\begin{tabular}[t]{l}AtlasNet\end{tabular}}}}%
    \put(0.88774749,0.3242339){\color[rgb]{0,0,0}\makebox(0,0)[lt]{\lineheight{1.25}\smash{\begin{tabular}[t]{l}Ours\end{tabular}}}}%
    \put(0,0){\includegraphics[width=\unitlength,page=1]{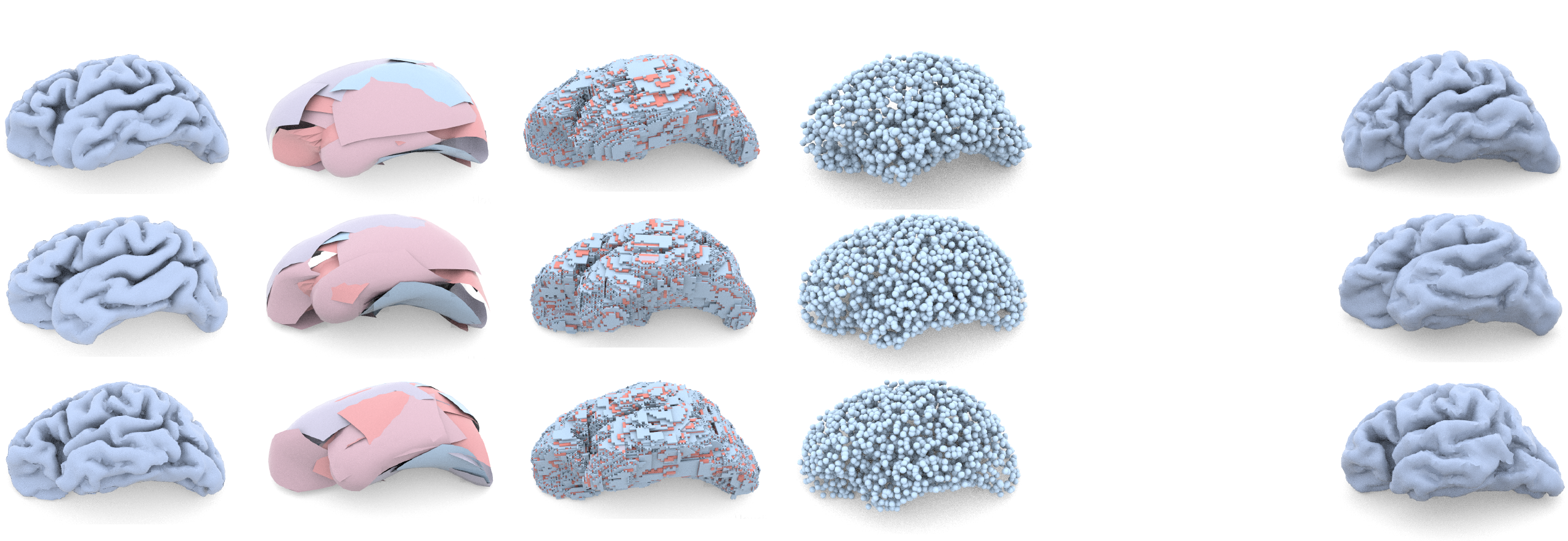}}%
    \put(0.00892201,0.31634027){\color[rgb]{0,0,0}\makebox(0,0)[lt]{\lineheight{1.25}\smash{\begin{tabular}[t]{l}(FreeSurfer)\end{tabular}}}}%
    \put(0,0){\includegraphics[width=\unitlength,page=2]{brain_comparison3.pdf}}%
    \put(0.70955585,0.32496479){\color[rgb]{0,0,0}\makebox(0,0)[lt]{\lineheight{1.25}\smash{\begin{tabular}[t]{l}Baseline\end{tabular}}}}%
    \put(0,0){\includegraphics[width=\unitlength,page=3]{brain_comparison3.pdf}}%
  \end{picture}%
\endgroup%

    \caption{Brain autoencoder. We build a curvature-to-curvature autoencoder and compare to the models based on point clouds, the AtlasNet \cite{groueix2018} (point clouds to surface) and the point-cloud AE \cite{Achlioptas2018LearningRA} (point clouds to point clouds), the voxel-based model OGN \cite{ogn2017} (IDs to voxels), and a mesh-based baseline model, which replaces the curvature in our model with vertex coordinates. All the neural networks, except for OGN, are trained on $1400$ cortical surfaces and validated on $200$ surfaces, which do not appear in the training set. Three of the predicted surfaces from the validation data are shown above. Although all the models can restore the brain structure in a large scale, only our model preserves the local fine structure. For more details see Section ~\ref{subsec:closed_surf_gen}.}
    \label{fig:brain1}
}

\maketitle
\begin{abstract}
This paper introduces a generative model for 3D surfaces based on a representation of shapes with mean curvature and metric, which are invariant under rigid transformation. Hence, compared with existing 3D machine learning frameworks, our model substantially reduces the influence of translation and rotation. In addition, the local structure of shapes will be more precisely captured, since the curvature is explicitly encoded in our model. Specifically, every surface is first conformally mapped to a canonical domain, such as a unit disk or a unit sphere. Then, it is represented by two functions: the mean curvature half-density and the vertex density, over this canonical domain. Assuming that input shapes follow a certain distribution in a latent space, we use the variational autoencoder to learn the latent space representation. After the learning, we can generate variations of shapes by randomly sampling the distribution in the latent space. Surfaces with triangular meshes can be reconstructed from the generated data by applying isotropic remeshing and spin transformation, which is given by Dirac equation. We demonstrate the effectiveness of our model on datasets of man-made and biological shapes and compare the results with other methods.


\end{abstract}  

\section{Introduction}
\begin{figure*}
	\centering
    \def\svgwidth{0.95\textwidth}
    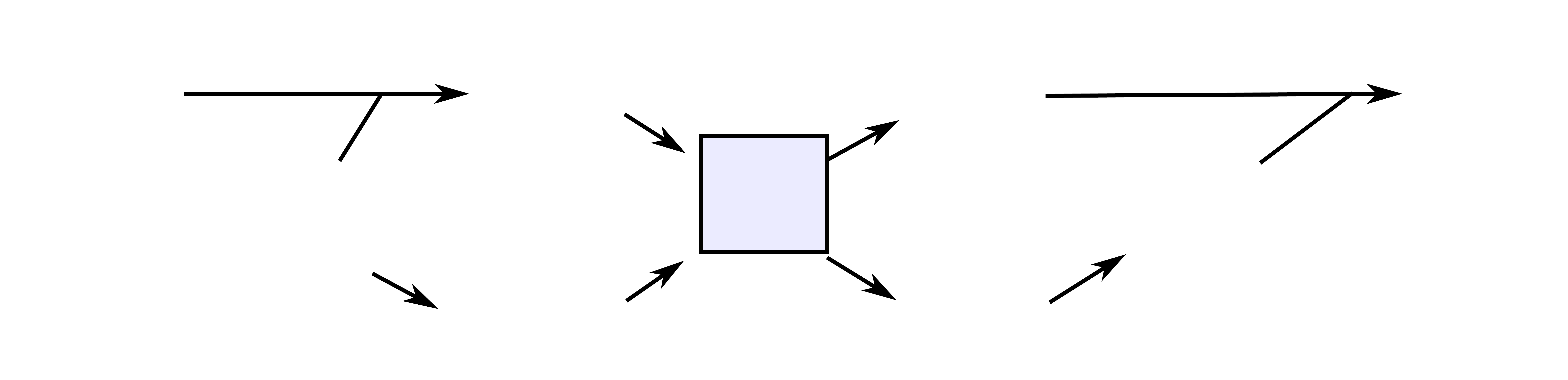
    \caption{The pipeline of our model for generating variant shapes: $(a)$ the conformal parameterization (Section \ref{subsec:conformal_param}), $(b)$ the density function extraction (Section \ref{subsec:rep}), $(c)$ the mean curvature half-density extraction (Section \ref{subsec:rep}), $(d)$ learning and generating (Section \ref{subsec:cnn}), $(e)$ the isotropic remeshing (Section \ref{subsec:rec_param}), $(f)$ solving the Dirac equation and applying the spin transformation (Section \ref{subsec:rec}).}
    \label{fig:pipeline}
\end{figure*}
While the convolutional neural network has achieved significant success in 2D image processing, more and more attention has recently been drawn to applying the technique to the domain of 3D shapes. Unlike 2D images, which are typically represented by a multidimensional tensor, the representation of 3D shapes is usually unstructured, hence the convolutional neural network is not directly applicable. Thus the main challenge is how to create a suitable representation for 3D shapes which can take advantage of the state-of-art machine learning frameworks. Several such representations based on point clouds \cite{Fan2017APS, Achlioptas2018LearningRA,groueix2018}, volumetric data \cite{ogn2017,Wang_2018,Wang2017OCNNOC}, and meshes \cite{Ben_Hamu_2018} have been proposed with different applications. However, all these representations are built on the positional data such as the coordinates of points, vertices or voxels.

In this paper, we propose a 3D deep generative model based on mean curvature and metric, which in discrete case are expressed by two functions that are invariant under Euclidean motion. It has the following advantages against the existing models:

\begin{wrapfigure}{r}{0.3\columnwidth}
\vspace{-0.8\intextsep}
\hspace*{-1\columnsep}
\includegraphics[scale=0.15]{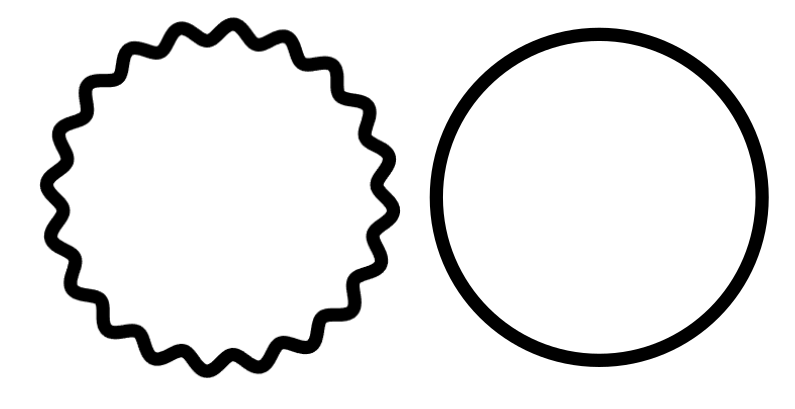}
\vspace{-1\intextsep}
\end{wrapfigure}

Firstly, our model preserves more detailed structure in case that the curvature plays a critical role, especially when the surface is highly folded and convoluted like the cortical surfaces in Figure~  \ref{fig:brain1}. The convolutional neural network (CNN) is known to be good at capturing not only the global features but also the local fine structure of data. Its effectiveness, however, relies on a proper distance function defined on the space of features. For example, the Euclidean distance between two vectors is a straightforward option. As the result, the bumpy circle (inset) will tend to be deformed through the neural network to the round circle, which is more regular and is close to the bumpy one under the measurement by Euclidean distance. In contrast, we adopt curvature representation and subsequently the distance between curvatures, by which two circles are clearly distinguishable, hence the small hills will be safely preserved. Secondly, our model is less affected by rigid transformation and uniform scaling. Thanks to the invariant quantities that constitute our representation and the CNN on sphere (see Section \ref{subsec:unaligned} for detailed discussion), we provide a simple and efficient way to handle the data without a consistent alignment.

The input shapes for our model are required to be surfaces with consistent simply-connected topology, e.g., the disk-like surface or the spherical surface. We first map the input surface to a canonical domain such as a sphere, where mean curvature and vertex density are extracted and recorded as the input data for the neural network. For generative models like VAE, the output is a variant of the input so it has the same form as the input. To reconstruct the shape, we first create a conformal parameterization by randomly sampling the points with respect to the generated density function and applying the isotropic remeshing. Then, we deform the mesh gradually towards the target shape with the prescribed mean curvature (see the attached videos).

A curvature-to-shape reconstruction algorithm with high accuracy is critical for generating plausible shapes. We follow the basic idea in \cite{crane_robust_2013} and \cite{Ye_2018}. The deformation between the domain and target shapes is given by the solution of the Dirac equation. We propose a modified equation with a larger solution space and it results in the reconstruction comparatively closer to the target shape. Furthermore, one might be concerned about the stability of curvature-based methods, since tiny errors in curvature might accumulate across the surface and significantly affect the final reconstruction. Indeed, in our case, previous methods fail to locally scale the shape correctly at regions with large curvature. In fact, it is hard to directly manipulate the local scaling with the Dirac equation. Therefore we design a new algorithm inspired by Chern et al. \shortcite{chern2015_close-to-conformal-deformations-of-volumes} to calibrate the area scaling factor.  This compensates for the shortcoming of the Dirac equation and significantly stabilize the reconstruction. 

We evaluate our reconstruction algorithm on several shapes, showing that our method outperforms previous methods visually and quantitatively. In addition to some preliminary applications such as shape remeshing, interpolation and clustering, we demonstrate randomly generated shapes from various datasets and compare to other 3D generative models.

In summary, the contribution of this paper is 1) an improved algorithm for shape reconstruction from curvature with area calibration and 2) a 3D shape deep learning framework based on curvature.

\begin{figure*}
\centering
\def\svgwidth{\textwidth}
\begingroup%
  \makeatletter%
  \providecommand\color[2][]{%
    \errmessage{(Inkscape) Color is used for the text in Inkscape, but the package 'color.sty' is not loaded}%
    \renewcommand\color[2][]{}%
  }%
  \providecommand\transparent[1]{%
    \errmessage{(Inkscape) Transparency is used (non-zero) for the text in Inkscape, but the package 'transparent.sty' is not loaded}%
    \renewcommand\transparent[1]{}%
  }%
  \providecommand\rotatebox[2]{#2}%
  \newcommand*\fsize{\dimexpr\f@size pt\relax}%
  \newcommand*\lineheight[1]{\fontsize{\fsize}{#1\fsize}\selectfont}%
  \ifx\svgwidth\undefined%
    \setlength{\unitlength}{596.64131225bp}%
    \ifx\svgscale\undefined%
      \relax%
    \else%
      \setlength{\unitlength}{\unitlength * \real{\svgscale}}%
    \fi%
  \else%
    \setlength{\unitlength}{\svgwidth}%
  \fi%
  \global\let\svgwidth\undefined%
  \global\let\svgscale\undefined%
  \makeatother%
  \begin{picture}(1,0.29583245)%
    \lineheight{1}%
    \setlength\tabcolsep{0pt}%
    \put(0,0){\includegraphics[width=\unitlength,page=1]{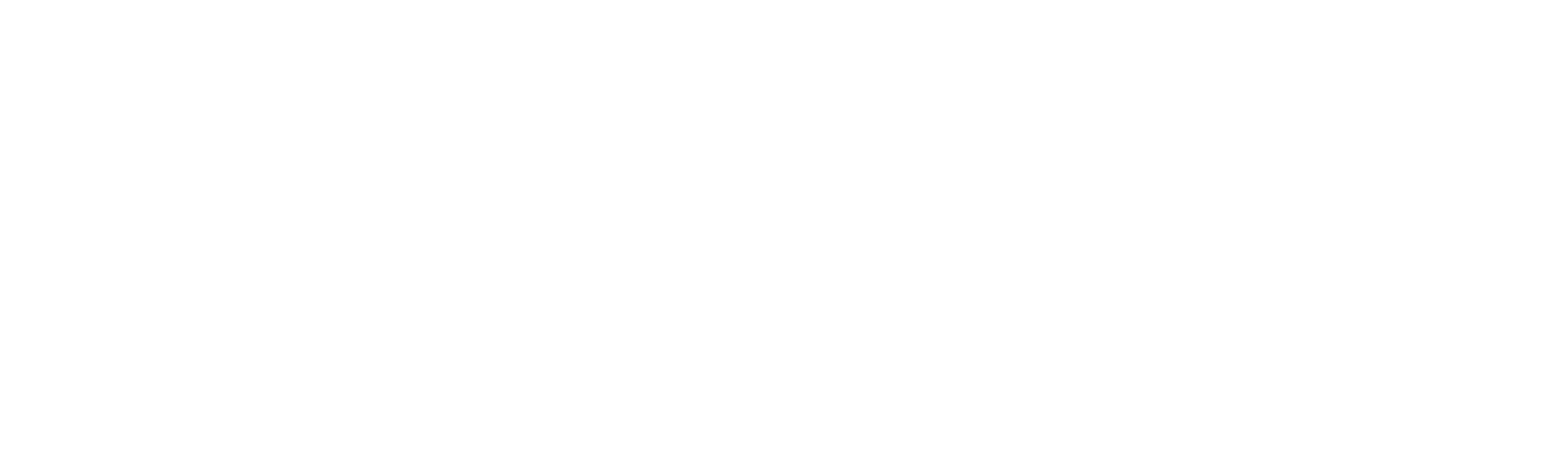}}%
    \put(0.9272152,0.23027368){\color[rgb]{0,0,0}\makebox(0,0)[lt]{\lineheight{1.25}\smash{\begin{tabular}[t]{l}$1.8$\end{tabular}}}}%
    \put(0.914645,0.07027902){\color[rgb]{0,0,0}\makebox(0,0)[lt]{\lineheight{1.25}\smash{\begin{tabular}[t]{l}$-1.2$\end{tabular}}}}%
    \put(0,0){\includegraphics[width=\unitlength,page=2]{animals.pdf}}%
    \put(0.8647145,0.2821438){\color[rgb]{0,0,0}\makebox(0,0)[lt]{\lineheight{1.25}\smash{\begin{tabular}[t]{l}\small mean curvature\end{tabular}}}}%
    \put(0.88087615,0.26687978){\color[rgb]{0,0,0}\makebox(0,0)[lt]{\lineheight{1.25}\smash{\begin{tabular}[t]{l}\small half-density\end{tabular}}}}%
  \end{picture}%
\endgroup%

\caption{The spherical conformal parameterizations of two animals are aligned by a M\"obius transformation with three landmark points. Then, they are packed into tensors with dimension $320\times 32\times 32 \times 2$. This figure shows a linear interpolation between the curvature representation of two shapes and the resulting shape reconstruction from the curvature representation.}\label{fig:intp}
\end{figure*}

\section{Related Work}
\subsection{Which invariant quantities determine an immersed surface in $\mathbb{R}^3$?}
\label{subsec:intrinsic}
It is well-known that an immersed surface in $\mathbb{R}^3$ is determined up to a Euclidean motion by its first and second fundamental forms. However, their representation depends on the choice of coordinate. Hence, in order to consistently represent 3D shapes based on the two fundamental forms, an identical triangulation for all shapes, which is not always possible, is required.

Other options are point-wise shape descriptors such as the heat kernel signature \cite{sun2009concise} and the wave kernel signature \cite{aubry2011wave}. Indeed, they have been employed in discriminative models for 3D shape classification and segmentation \cite{doi:10.1111/cgf.12693}. But they can hardly be used for generative models, because it is unclear whether these shape descriptors completely determine the shapes or how to reconstruct shapes from them.

The idea of this paper comes originally from 
Bonnet \cite{bonnet1967memoire}. In fact, except for some very special cases, an immersed surface is completely determined by  conformal structure, regular homotopy class and mean curvature half-density, which is a scale-independent variant of the mean curvature \cite{kamberov1998prescribing}. The exceptions, called the Bonnet immersions, includes minimal surfaces, constant mean curvature surfaces and Bonnet pairs. In our case the regular homotopy class is unnecessary, since we only consider the simply-connected surfaces which have only one unique regular homotopy class \cite{pinkall1985regular}. In summary, generic simply-connected immersed surfaces are uniquely determined by the conformal structure and the mean curvature half-density.

\subsection{Quaternions, Dirac-type operators and Spin Transformation}
\label{subsec:dirac}
Now, we sketch the idea how to construct a surface from the mean curvature half-density. Roughly speaking, for every point on the surface we rotate its infinitesimal neighbourhood with a quaternion. Recall that a quaternion is a 4-dimensional vector $q = a + bi + cj+ dj$ with the multiplicative structure:
\begin{align*}
i^2&=j^2=k^2 = -1,\\
ij=-ji = k,& \quad jk = -kj = i,\quad ki = -ik = j.
\end{align*}
We always identify vectors in $\mathbb{R}^3$ as pure imaginary quaternions 
\[(x,y,z) \mapsto xi + yj +zk.\]
Any quaternion can be written as $q = \lvert q\rvert(\cos \frac{\theta}{2} + \sin \frac{\theta}{2}u)$, where $\theta \in [0,2\pi)$ and $u\in \mathbb{R}^3 \subset \mathbb{H}$. It is well known that $q$ gives a scale rotation in $\mathbb{R}^3$ with scaling factor $\lvert q\rvert^2$,   rotation angle $\theta$ and rotation axis $u$. The rotation is given by
\[R_q(v) =  \overline{q}\cdot v \cdot q.\]
The explicit construction of shapes from mean curvature half-density and conformal structure is called spin transformation.  Suppose given an immersion of a surface $f:M\rightarrow \mathbb{R}^3$ and a quaternion-valued function on the surface $\phi: M\rightarrow \mathbb{H}$, which is understood as a continuously varying rotation at each point. We scale and rotate every tangent plane by
\begin{equation}
\label{eqn:spin_tr}
d \tilde{f} = \overline{\phi}\cdot df\cdot\phi.
\end{equation}
However, there is no guarantee that these rotated tangent planes will again form a surface. For simply connected surface, $d \tilde{f}$ is again the tangent plane of an immersion of surface if and only if it is closed:
\[dd\tilde{f} = 0\]
It turns out to be equivalent to the Dirac equation  \cite{kamberov1998}
\begin{equation}
\label{eqn:dirac_eq}
D_f \phi = \rho \phi,
\end{equation}
where the Dirac operator is defined by
\begin{equation}
\label{eqn:dirac}
D_f \phi = -\frac{df\wedge d\phi}{\lvert df\rvert^2},
\end{equation}
and $\rho:M\rightarrow \mathbb{R}$ is a real-valued function. Therefore, any solution of the equation \eqref{eqn:dirac_eq} will induce a new immersion $\tilde{f}:M\rightarrow \mathbb{R}^3$ by $\tilde{f} = \int_M d\tilde{f}$.
Moreover, the mean curvature $\tilde{H}$ of $\tilde{f}$ is given by
\begin{equation}
\label{eqn:mchd_change}
\tilde{H} \lvert d \tilde{f} \rvert = H \lvert df \rvert + \rho \lvert df\rvert, 
\end{equation}
where $H$ is the mean curvature of the original surface $f$. Observe that, due to the scaling factor $\lvert d \tilde{f} \rvert$ in \eqref{eqn:mchd_change}, one can not fully control the mean curvature $\tilde{H}$. However, by introducing a variant notion, namely the mean curvature half-density: 
\begin{equation}
h := H \lvert df\rvert, 
\end{equation}
the equation \eqref{eqn:mchd_change} turns to
\begin{equation}
\tilde{h} = h + \rho \lvert df\rvert.
\end{equation}  
This means that the mean curvature half-density $\tilde{h}$ can be precisely realized as long as the solution $\phi$ for equation \eqref{eqn:dirac_eq} exists.

Crane et al. \cite{crane_spin_2011} first discretize the equation \eqref{eqn:dirac} and show applications in computer graphics, such as curvature painting. The following works are, e.g.,  Crane et al. \cite{crane_robust_2013} use the spin transformation for surface fairing. Liu et al. \cite{liu2017_a-dirac-operator-for-extrinsic-shape} construct a continuous spectrum of operators between the square of the Dirac operator and the Laplace-Beltrami operator. These operators are utilized to enhance surface matching and segmentation problems. Ye et al.  \cite{Ye_2018} create a framework, which consistently discretized the extrinsic Dirac operator and an intrinsic Dirac operator. In this paper, we improve the reconstruction based on \cite{crane_spin_2011,Ye_2018}  by solving an equation with a larger solution space and introducing an area calibration (see Section \ref{subsec:rec}).

\subsection{Deep Generative Models for 3D shapes}
Various representations of surfaces have been proposed for 3D shape generation, e.g., models based on volumetric representation \cite{3dgan,ogn2017,Smith2017ImprovedAS,Wang2017OCNNOC,Wang_2018},  or point clouds representation \cite{Fan2017APS, Nash_2017, Achlioptas2018LearningRA}. These methods are particularly applicable for the dataset with inconsistent topology. However, without knowing the mesh structure it is hard to capture the fine structure of certain highly complicated surfaces (see Figure \ref{fig:brain1}).

Our model is closer to the following works, which take the mesh structure into account. Ben-Hamu et al.  \shortcite{Ben_Hamu_2018} propose a representation based on multiple charts, which conformally map different parts of shapes to a domain. Since features over each chart are normalized separately, the fine structure will be better preserved than with a single chart. However, while the creation of such charts requires a sparse correspondence, reconstruction of shapes from the charts needs a template shape, which amounts to a dense correspondence. In order to find such correspondence, one has to introduce a time-consuming workflow beforehand. Groueix et al. \shortcite{groueix2018} learns a parameterization of shapes with multiple embedded charts. Hence one does not have to manually create the charts. However, the generated charts do not always perfectly fit with each other, nor do they preserve as much details as the ones in \cite{Ben_Hamu_2018}. Umetani \shortcite{Umetani_2017} develops a depth map representation with a cube as the domain. This representation works well for close-to-convex shapes like cars, but would be difficult to be applied on highly curved and non-convex shapes. Kostrikov et al. \shortcite{Kostrikov_2018_CVPR} use the same Dirac operator as ours. But they merely replaced the Laplace-Beltrami operator in the neural network with the Dirac operator, thus the real power of the Dirac operator, namely its connection to conformal transformation, is not exploited.

\section{Method}
The main pipeline of our model is depicted in Figure \ref{fig:pipeline}. In the sequel, we will explain the detailed methods for encoding shapes with curvature and vertex density in Section \ref{subsec:conformal_param} and \ref{subsec:rep}, building a neural network based on our representation in Section \ref{subsec:cnn} and reconstruction of shapes in Section \ref{subsec:rec_param} and \ref{subsec:rec}.

\paragraph*{Encoding the Conformal Structure}
In discrete case, how to encode a shape in the scheme of the Bonnet problem (Section \ref{subsec:intrinsic})?  While the mean curvature half-density can be represented by a vertex-based or face-based function, it is not straightforward to pack the conformal structure in a form that is suitable for machine learning pipeline. For example, we can recover the shape of a cow from its spherical conformal parameterization ((b) in Figure \ref{fig:conf_map}) by prescribing the function of mean curvature half-density ((c) in Figure \ref{fig:conf_map}). But it is not clear how to represent a spherical mesh that is conformal equivalent to a given shape purely by scalar functions. One might consider the notion of discrete conformal equivalence for triangular meshes by length cross-ratio on edges  (\cite{springborn2008conformal}). But it is unclear how to transfer the length cross-ratio across different meshes. 

\begin{figure}[H]
\centering
\def\svgwidth{1\columnwidth}
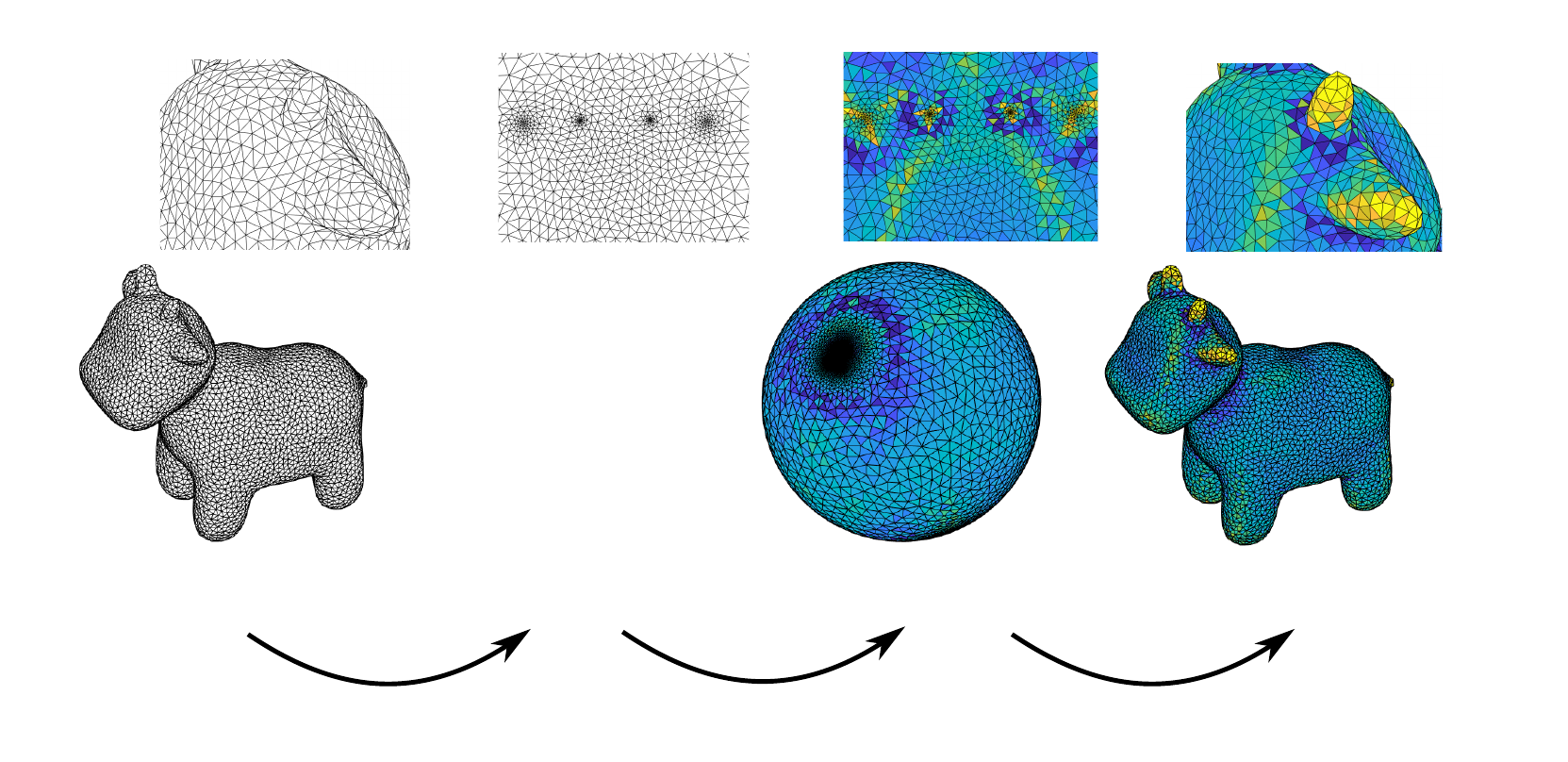
\caption{\cite{Ye_2018} shows that a simply-connected surface in $\mathbb{R}^3$ can be faithfully reconstructed from its conformal parameterization by prescribing the mean curvature half-density.}\label{fig:conf_map}
\end{figure}

Recall that the conformal structure is the set of metrics modulo the equivalence relation $g\sim e^{2u}g$, i.e., two metrics are identified if they only differ by a scaling at each point. Therefore, instead of encoding the conformal structure, we encode the metric of shapes. In general, the space of all metrics still does not have an efficient form of representation, thus we focus on a smaller subset, i.e., the isotropic meshing. Since the conformal map is locally isotropic, i.e., it takes an isotropic mesh to a close-to-isotropic  mesh (see the zoom-in in Figure \ref{fig:conf_map}), and we know that the isotropic meshing is usually generated by the centroidal Voronoi tessellation (CVT) with respect to a density function \cite{alliez2003isotropic}, this density function can be utilized as an approximation of a metric. Therefore, at the beginning of our pipeline all the input shapes are isotropically remeshed (like $(a)$ in Figure \ref{fig:conf_map}). Then, we successively take the following procedures.

\subsection{Conformal parameterization}
\label{subsec:conformal_param}
We map all the shapes to a canonical domain, e.g., the unit disk for disk-like surfaces and the unit sphere for spherical surfaces. The resulting disk-like or spherical meshes are called the conformal parameterization. However, these maps are not unique but differ by a conformal automorphism of the domain. To deal with the ambiguity one may choose from the following approaches depending on the application:

\paragraph*{Landmark alignment} We know that the conformal automorphism of  $S^2$, i.e., the M\"obius transformation, is fully determined by three distinguished points and the conformal automorphism of a disk is determined by one point and one rotation. Hence we choose two landmark points for disk-like surfaces and three landmark points for closed surfaces and align these landmarks via  a conformal mapping. One example is shown in Figure \ref{fig:alignment}.
\paragraph*{Landmark-free alignment} For example, \cite{doi:10.1111/cgf.13503} proposed a canonical M\"obius transformation such that the mass center is aligned with the sphere center. Then, we register two spherical meshes of centered M\"obius transformations by searching for an  optimal rotation.
\paragraph*{Without any alignment at all} This will result in a larger shape latent space and consequently poses higher demands on the capacity of neural network, because, for example, a rotation of shapes might also cause a rotation of curvature function. However, our model is particularly good at capturing this uncertainty (see the discussion in Section \ref{subsec:unaligned}).

Specifically, there are many available algorithms for conformal parameterization for disk-like and spherical surfaces, e.g.,  \cite{Gu_2004,crane_robust_2013, Choi_2015, Choi_20152, Sawhney:2017:BFF, Ye_2018}. In fact, we did not observe significant differences between these algorithms in our experiments.

\subsection{Making the representation}
\label{subsec:rep}
In order to build the neural network, we need some fixed meshes for canonical domains. In particular, we use the standard $256\times 256$ grids for the disk. For the spherical domain, we obtain a spherical mesh by iteratively applying the 1-to-4 subdivision and normalization on an icosahedron. 

Then, we interpolate the following two functions from the conformal parameterization of shapes to the domain with inverse distance weight.

\paragraph*{Mean curvature half-density} The mean curvature half-density $h$ is a face-based function given by \cite{Ye_2018}
\begin{equation}
\label{eqn:dis_mchd}
h_i = \frac{\sum_j \lvert e_{ij}\rvert \tan \theta_{ij}/2}{2 \sqrt{A_i}},
\end{equation}
where the sum runs over all the edges $e_{ij}$ of the face $T_i$, $\theta_{ij}$ are bending angles at the edge $e_{ij}$ and $A_i$ are the face area.

\paragraph*{Vertex density function} We estimate the density function $\mathfrak{d}$ by the reciprocal of vertex area, $\mathfrak{d}_i:= 1 /\tilde{A}_i$, where $\tilde{A}_i$ is the vertex area of the conformal parameterization. We do not normalize the density $\mathfrak{d}$, since the integral of the piecewise constant function $\int_U \mathfrak{d} dA = \sum_i \mathfrak{d}_i \tilde{A}_i = i$ is equal to the number of points located in the area $U$. At the step of  reconstruction, this gives us the information about how many points should be sampled. In the experiment, we observe that the logarithmic density $\tilde{\mathfrak{d}}:= \log \mathfrak{d}$ is more evenly distributed. Therefore, the logarithmic density $\tilde{\mathfrak{d}}$ is instead recorded on the domain. 
 
\subsection{Building CNNs over meshes}
\label{subsec:cnn}
Since the disk-like surface is represented like a 2D image with two channels, any classical CNNs can be directly applied. Hence, we will focus on the case of spherical surfaces.

\paragraph*{Convolution layers} Many CNNs on arbitrary graphs or surfaces have been proposed in recent works, see \cite{Bruna2014SpectralNA,kipf2016semi, Cheb_2016, Monti_2017, doi:10.1111/cgf.12693,MasBosBroVan15,Maron17} and the survey \cite{7974879}. We opt for a mesh-based CNN based on small tangent patches \cite{Tat18}.

\begin{wrapfigure}{R}{0.8in}
	\vspace{-0.5\intextsep}
	\hspace*{-.75\columnsep}
	\includegraphics[width=1in]{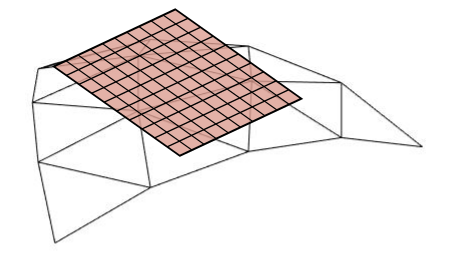}
	\vspace{-0.5\intextsep}
	
\end{wrapfigure}

Each face of the domain is assigned with a tangent plane, identified with $ \mathbb{R}^2$, at the barycenter. Let $l$ be a positive number such that the projection of the triangular face lies entirely in the patch $[-l,l]\times [-l,l]$ on the tangent plane. This projection $\pi$ gives a local coordinate system of the points in the pre-image $\pi^{-1}([-l,l]\times [-l,l]) \subset S^2$. Hence, the functions restricted in this region can be interpolated to some grids on the patch. The distortion caused by the projection is neglectable when the size of the patches is small. We choose a fixed length $l$ such that all the triangular faces on the domain are projected inside the corresponding patches.  The convolution is the ordinary $2D$ convolution within each patch with the shared ﬁlter weights across different patches.

\begin{figure}[h!]
\centering
\def\svgwidth{0.8\columnwidth}
\begingroup%
  \makeatletter%
  \providecommand\color[2][]{%
    \errmessage{(Inkscape) Color is used for the text in Inkscape, but the package 'color.sty' is not loaded}%
    \renewcommand\color[2][]{}%
  }%
  \providecommand\transparent[1]{%
    \errmessage{(Inkscape) Transparency is used (non-zero) for the text in Inkscape, but the package 'transparent.sty' is not loaded}%
    \renewcommand\transparent[1]{}%
  }%
  \providecommand\rotatebox[2]{#2}%
  \newcommand*\fsize{\dimexpr\f@size pt\relax}%
  \newcommand*\lineheight[1]{\fontsize{\fsize}{#1\fsize}\selectfont}%
  \ifx\svgwidth\undefined%
    \setlength{\unitlength}{1525.76697642bp}%
    \ifx\svgscale\undefined%
      \relax%
    \else%
      \setlength{\unitlength}{\unitlength * \real{\svgscale}}%
    \fi%
  \else%
    \setlength{\unitlength}{\svgwidth}%
  \fi%
  \global\let\svgwidth\undefined%
  \global\let\svgscale\undefined%
  \makeatother%
  \begin{picture}(1,0.2082927)%
    \lineheight{1}%
    \setlength\tabcolsep{0pt}%
    \put(0,0){\includegraphics[width=\unitlength,page=1]{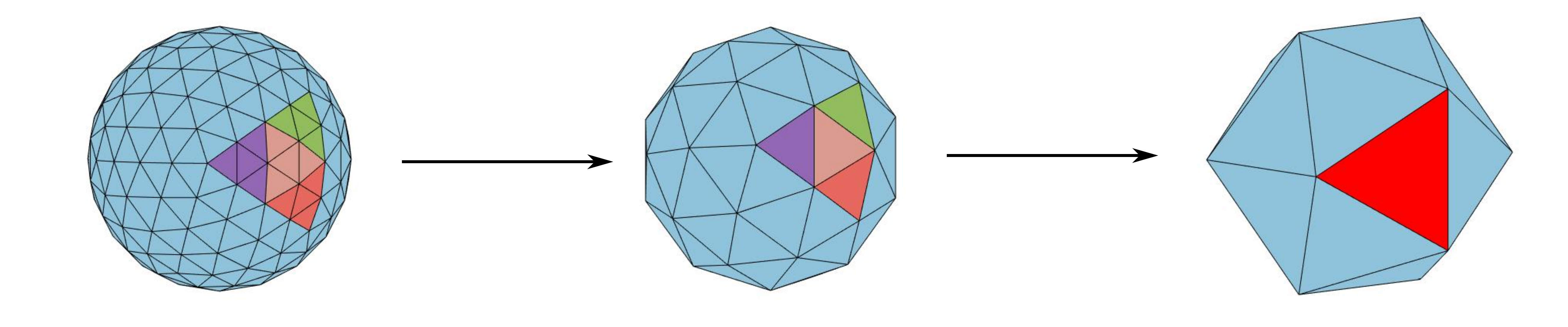}}%
    \put(0.23771605,0.11762664){\color[rgb]{0,0,0}\makebox(0,0)[lt]{\lineheight{1.25}\smash{\begin{tabular}[t]{l}D.S. 1\end{tabular}}}}%
    \put(0.58773118,0.12061589){\color[rgb]{0,0,0}\makebox(0,0)[lt]{\lineheight{1.25}\smash{\begin{tabular}[t]{l}D.S. 2\end{tabular}}}}%
  \end{picture}%
\endgroup%

\caption{Downsampling layers based on the subdivision structure of the spherical meshes. The tensors in the previous layer, which are corresponding to a common triangle in the next layer, are merged to the tensor  associated with the father triangle. These downsampling layers respect the spatial relations among the triangles.}\label{fig:downsampling} 
\end{figure}

\paragraph*{Downsampling and Upsampling layers}
Like the MaxPooling and UpPooling layers for classical CNNs, we need the same sort of operations for mesh domain to decrease and increase the spatial dimension of neural network. One can first apply the ordinary $2D$ pooling layers within each patch. Furthermore, since our spherical domain is constructed by subdividing an icosahedron, it is naturally endowed with a hierarchical structure (Figure \ref{fig:downsampling}), which gives rise to downsampling and upsampling layers between spherical meshes with different refinements.

The detailed architectures of our convolutional neural networks are depicted in the appendix.
 
\subsection{Reconstruction of Conformal Parameterization}
\label{subsec:rec_param}
In order to construct a conformal parameterization  from a given vertex density function $\mathfrak{d}$, we first randomly sample $n_i$ points in every faces of the domain, where $n_i = \mathfrak{d}_i \tilde{A}_i$ and $\tilde{A}_i$ is the face area. Next, an isotopic meshing is constructed as follows.

\paragraph*{Centroidal Voronoi Tessellation}

The isotropic meshing is usually made by centroidal Voronoi tessellation \cite{du1999centroidal}. Given a set of points $\{v_i\}$ in a metric space, particularly $\mathbb{R}^2$ or $S^2$. The Voronoi region $V_i$ corresponding to $v_i$ is defined by
\begin{equation}
\label{eqn:voronoi}
V_i = \{x| \lvert x-v_i\rvert \leq \lvert x-v_j\rvert, \, j\neq i\},
\end{equation}
which are polygons (see Appendix \ref{app:CVT} for the formula for computing the weighted centroid of polygons). Given a density function $\mathfrak{d}$, the centroid $v^*_i$ of the polygon $V_i$ is given by
\begin{equation}
\label{eqn:centroid}
 v^* = \frac{\int_ V y \mathfrak{d}(y) dy}{\int_V  \mathfrak{d}(y) dy}.
\end{equation}

We call a point set $\{v_i\}$ the weighted centroidal Voronoi tesselation if $v_i = v_i^*$ holds true for all $i$.

\begin{figure}[h]
\centering
\def\svgwidth{1\columnwidth}
\begingroup%
  \makeatletter%
  \providecommand\color[2][]{%
    \errmessage{(Inkscape) Color is used for the text in Inkscape, but the package 'color.sty' is not loaded}%
    \renewcommand\color[2][]{}%
  }%
  \providecommand\transparent[1]{%
    \errmessage{(Inkscape) Transparency is used (non-zero) for the text in Inkscape, but the package 'transparent.sty' is not loaded}%
    \renewcommand\transparent[1]{}%
  }%
  \providecommand\rotatebox[2]{#2}%
  \newcommand*\fsize{\dimexpr\f@size pt\relax}%
  \newcommand*\lineheight[1]{\fontsize{\fsize}{#1\fsize}\selectfont}%
  \ifx\svgwidth\undefined%
    \setlength{\unitlength}{495.47980842bp}%
    \ifx\svgscale\undefined%
      \relax%
    \else%
      \setlength{\unitlength}{\unitlength * \real{\svgscale}}%
    \fi%
  \else%
    \setlength{\unitlength}{\svgwidth}%
  \fi%
  \global\let\svgwidth\undefined%
  \global\let\svgscale\undefined%
  \makeatother%
  \begin{picture}(1,0.22018347)%
    \lineheight{1}%
    \setlength\tabcolsep{0pt}%
    \put(0,0){\includegraphics[width=\unitlength,page=1]{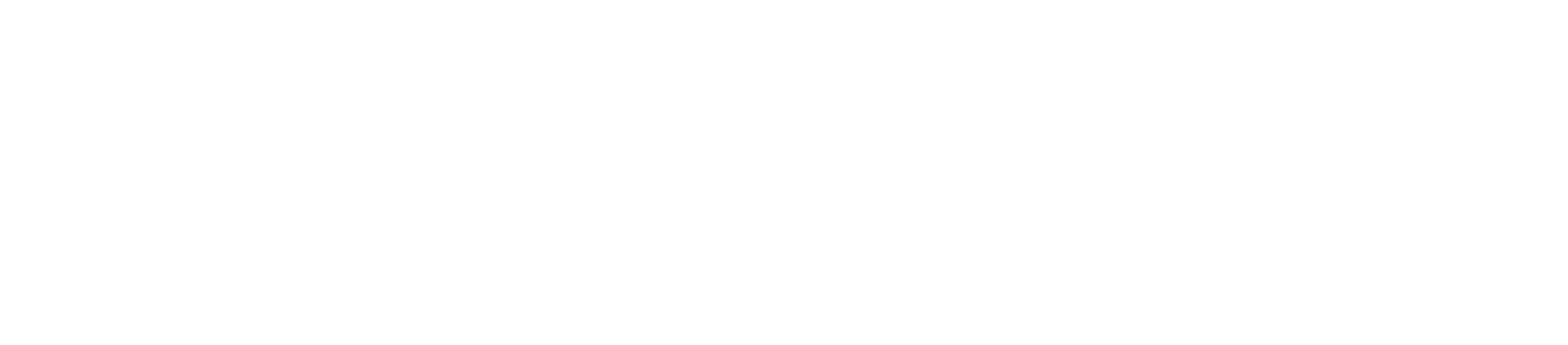}}%
    \put(0.02348696,0.02074988){\color[rgb]{0,0,0}\makebox(0,0)[lt]{\lineheight{1.25}\smash{\begin{tabular}[t]{l}Sampling\end{tabular}}}}%
    \put(0.20223567,0.01853109){\color[rgb]{0,0,0}\makebox(0,0)[lt]{\lineheight{1.25}\smash{\begin{tabular}[t]{l}Voronoi D.\end{tabular}}}}%
    \put(0.39202272,0.01853192){\color[rgb]{0,0,0}\makebox(0,0)[lt]{\lineheight{1.25}\smash{\begin{tabular}[t]{l}1-st iter.\end{tabular}}}}%
    \put(0.5743223,0.01853192){\color[rgb]{0,0,0}\makebox(0,0)[lt]{\lineheight{1.25}\smash{\begin{tabular}[t]{l}5-th iter.\end{tabular}}}}%
    \put(0.71397561,0.02074988){\color[rgb]{0,0,0}\makebox(0,0)[lt]{\lineheight{1.25}\smash{\begin{tabular}[t]{l}Delaunay Tri.\end{tabular}}}}%
    \put(0,0){\includegraphics[width=\unitlength,page=2]{cvt.pdf}}%
  \end{picture}%
\endgroup%

\caption{Centroidal Voronoi Tessellation. In order to obtain an isotropic meshing with respect to a given density, we first sample a point set according to the density and repeatedly apply the Lloyd's relaxation. Observe that the point set becomes more and more isotropic as the iteration goes.}
\label{fig:cvt}
\end{figure}

In this paper we use Lloyd relaxation to compute the CVT. Given a point set $\{v_i\}$ we iteratively update the point $v_i$ with the corresponding centroid $v_i^*$ until it converges 
(see Figure \ref{fig:cvt}):
\begin{enumerate}
\item Randomly sample the points with respect to the density $\mathfrak{d}$ (defined in Section \ref{subsec:rep}).
\item Create the Voronoi diagram. For the disk case, we have to be a bit careful that the Voronoi cells close to the boundary are mostly unbounded. Hence we reflect the points close to the boundary, so that all the Voronoi cells inside or close to the unit disk are bounded. 
\item Compute the weighted centroids of the (bounded) Voronoi cells and, for the disk case, remove the points lying outside the disk (see Figure \ref{fig:ccvt}). 
\end{enumerate}

\begin{figure}[h]
\centering
\def\svgwidth{1\columnwidth}
\begingroup%
  \makeatletter%
  \providecommand\color[2][]{%
    \errmessage{(Inkscape) Color is used for the text in Inkscape, but the package 'color.sty' is not loaded}%
    \renewcommand\color[2][]{}%
  }%
  \providecommand\transparent[1]{%
    \errmessage{(Inkscape) Transparency is used (non-zero) for the text in Inkscape, but the package 'transparent.sty' is not loaded}%
    \renewcommand\transparent[1]{}%
  }%
  \providecommand\rotatebox[2]{#2}%
  \newcommand*\fsize{\dimexpr\f@size pt\relax}%
  \newcommand*\lineheight[1]{\fontsize{\fsize}{#1\fsize}\selectfont}%
  \ifx\svgwidth\undefined%
    \setlength{\unitlength}{544.72472664bp}%
    \ifx\svgscale\undefined%
      \relax%
    \else%
      \setlength{\unitlength}{\unitlength * \real{\svgscale}}%
    \fi%
  \else%
    \setlength{\unitlength}{\svgwidth}%
  \fi%
  \global\let\svgwidth\undefined%
  \global\let\svgscale\undefined%
  \makeatother%
  \begin{picture}(1,0.27955494)%
    \lineheight{1}%
    \setlength\tabcolsep{0pt}%
    \put(0,0){\includegraphics[width=\unitlength,page=1]{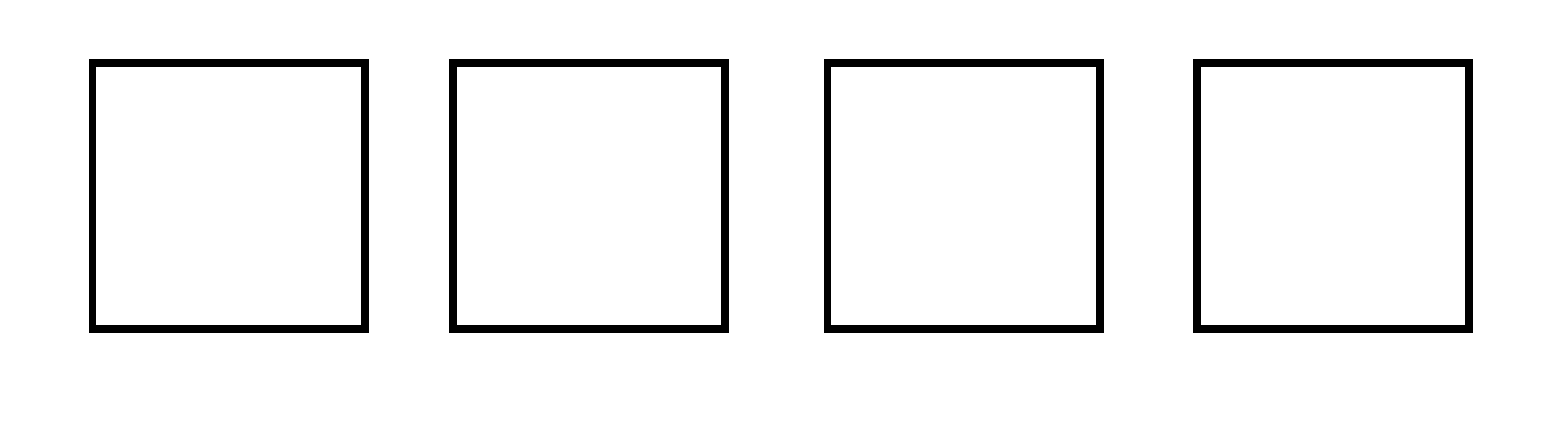}}%
    \put(0.03705011,0.03352662){\color[rgb]{0,0,0}\makebox(0,0)[lt]{\lineheight{1.25}\smash{\begin{tabular}[t]{l}Voronoi D.\end{tabular}}}}%
    \put(0.29176598,0.03702372){\color[rgb]{0,0,0}\makebox(0,0)[lt]{\lineheight{1.25}\smash{\begin{tabular}[t]{l}Flipping\end{tabular}}}}%
    \put(0.53287402,0.03702227){\color[rgb]{0,0,0}\makebox(0,0)[lt]{\lineheight{1.25}\smash{\begin{tabular}[t]{l}Update\end{tabular}}}}%
    \put(0.77354164,0.03352662){\color[rgb]{0,0,0}\makebox(0,0)[lt]{\lineheight{1.25}\smash{\begin{tabular}[t]{l}Remove\end{tabular}}}}%
    \put(0,0){\includegraphics[width=\unitlength,page=2]{figs/ccvt/ccvt.pdf}}%
  \end{picture}%
\endgroup%

\caption{Constraint CVT. To avoid dealing with unbounded Voronoi cells, we flip the points, which are close to the boundary, such that the cells close to the boundary are all bounded.}\label{fig:ccvt}
\end{figure}

Then, a Delaunay triangulation is constructed by taking the dual of the Voronoi diagram. Generally, this triangulation does not perfectly fit the disk at the boundary, but it does not significantly affect the global appearance of shapes.

\subsection{Surface Reconstruction}
\label{subsec:rec}
Now, we are ready to reconstruct the surface from a conformal parameterization with prescribed mean curvature half-density. In the following we first demonstrate an improved reconstruction method which is a slight modification of \cite{Ye_2018} and then introduce a new procedure of area calibration, which would be particularly effective when the area scaling is not accurately restored by the previous method. 

\paragraph*{Dirac Energy}
In practice, the exact solution of the Dirac equation \eqref{eqn:dirac_eq} can hardly be obtained, so we actually search for the solution $\phi:M\rightarrow \mathbb{H}$ such that:
\[(D_f -\rho - \sigma) \phi = 0\]
for a very small real number $\sigma$, which actually amounts to the eigenvalue problem 
\[(D_f -\rho) \phi = \lambda \phi,\]
where $\lambda$ is the eigenvalue with the smallest magnitude \cite{crane_spin_2011}. 

In discrete case, $D_f -\rho$ is a $\lvert F\rvert \times \lvert F\rvert$ quaternion-valued matrix \cite{crane_spin_2011}, or in practice, a $4\lvert F\rvert \times 4\lvert F\rvert$ real-valued matrix such that any quaternion $q = a + b i +c j +d k$ is represented by a $4\times 4$ real-valued block matrix:
\[\begin{pmatrix} a & -b & -c & d \\ b & a & -d & c \\ c & d & a & -b \\ d & -c & b & a\end{pmatrix}.\]
 
We briefly introduce the discretization of the matrix $D_f -\rho$ and refer the reader to \cite{crane_spin_2011, Ye_2018} for more details. Let $e_{ij} \in \mathrm{Im}(\mathbb{H})$ be the oriented edge embedded in the quaternion space and $\mathbf{H}_{ij} := \frac{1}{2} \lvert e_{ij} \rvert \tan \frac{\theta_{ij}}{2}$ be the integrated mean curvature at the edge $e_{ij}$, where $\theta_{ij}$ is the bending angle between the face $i$ and $j$. The matrix of the Dirac operator is a $4 \lvert F\rvert \times 4\lvert F\rvert$ matrix $D_f$ given by (\cite{Ye_2018})
\[(D_f \phi)_i = \frac{1}{2}E_{ij}\cdot \phi_j - \mathbf{H}_i \phi_i, \]
where $E_{ij} : = 2\mathbf{H}_{ij} + e_{ij}$ and $\mathbf{H}_i = \sum_j \mathbf{H}_{ij}$. The discrete form of $\rho$ is a $4\lvert F\rvert \times 4 \lvert F\rvert$ diagonal matrix $P$ with the discrete mean curvature half-density \eqref{eqn:dis_mchd} as the diagonal. Instead of building the target shape in one step, we slowly flow the initial shape to the target for the purpose of stability. Hence, we build the matrix $\hat{D}(t) = D_f - tP$, where $t \in [0,1]$ is a step length parameter.

We observe that, even though this face-based Dirac operator gives the exact solution, it is not numerically stable, because its solution space is often too large (technically, some solutions that give the edge-constraint normals far from the actual face normal will result in unwanted transformations). On the other hand, while the vertex-based operators in \cite{crane_spin_2011, Ye_2018} works well in many cases, they are not able to faithfully recover the high curvature regions on the surface, because their solution spaces are too limited. To have a balance between these two approaches we propose the following regularized energy based on the face-based operator:

\[E_D(t) =  \hat{D}^T(t) \cdot \hat{D}(t) + cR,\]
where $c$ is a positive coefficient and $R$ is the $4\lvert F\rvert \times 4\lvert F\rvert$ regularization matrix such that
\[R = \sum_{ij}\lvert e_{ij}^*\rvert (\phi_i-\phi_j)^2,\]
where the sum runs over all adjacent faces $i$ and $j$. Note that the weights with the dual edge length are used in \cite{Chern:2018:SFM}. To have  finer control of the regularizer, one can decompose $R$ into four components and set different weights as in \cite{Chern:2018:SFM}, but we did not see that this will make any obvious difference in our setting. Empirically, the coefficient $c$ is set to be $0.001\max\limits_{ij}\lvert e_{ij}\lvert$.

By the min-max principle, solving the generalized eigenvalue problem 
\[E_D(t) \phi = \lambda M \phi,\]
where $M$ is the mass matrix, is actually equivalent to minimizing the energy
\[ \min E_D, \text{ s.t. } \lvert \phi \rvert = 1,\]
with the metric defined by $\lvert \phi \rvert^2 := \phi^T \cdot M \cdot \phi$.

Finally, the edges are constructed by the spin transformation
\[e_{ij} \mapsto \mathrm{Im}( \overline{\phi_i} \cdot E_{ij} \cdot \phi_j ),\]
the position of vertices $v_i$ are recovered by solving the Poisson equation (see Section 3 of \cite{sorkine2007rigid} or Section 5.6 of \cite{crane_spin_2011}). In the attached videos, we prescribe the mean curvature half-density of two shapes (red) on their conformal parameterization (blue) and it shows deformation from the sphere to the original shapes.

\paragraph*{Area calibration}
Even though the Dirac operator with regularization term improves the accuracy of reconstruction, we observe that some area distortion is still visible, especially at the region with really high curvature. To overcome this problem, we make the reconstruction algorithm be aware of the area scaling factor. Chern et al. \cite{chern2015_close-to-conformal-deformations-of-volumes} prescribe a volumetric scaling factor $e^u$ and obtains the close-to-conformal volumetric  deformation by minimizing an energy $E_u$ depending on $u$. While the energy $E_u$ in \cite{chern2015_close-to-conformal-deformations-of-volumes} is specifically designed for 3D volumetric meshes, an analogy for 2D surfaces still holds in smooth case:

\begin{theorem}
\label{thm:closing}
Let $f:M \rightarrow \mathbb{R}^3 \subset \mathbb{H}$ be an isometric immersion and $h: M\rightarrow \mathbb{R}$ be any function. The quaternion gradient is defined by
\[\grad_f h = df(\grad h).\]
The spin transformation $d\tilde{f} := \overline{\phi}\cdot df \cdot \phi$ with $D_f \phi =0$ is closing if
\begin{equation}
\label{eqn:closing_area}
d\phi \phi^{-1}= -\frac{1}{2} G df,
\end{equation}
where $G:= \grad_f u$ is the gradient of the logarithmic factor $e^u:= \lvert \phi \rvert^2$.
\end{theorem}
\begin{proof}
See Appendix \ref{app:closing}.
\end{proof}

Therefore, given a spin transformation induced from $\phi$ with the area factor $u = \log \lvert \phi\rvert$, the quaternion-valued $1$-form 
\[\omega := d\phi + \frac{1}{2}Gdf\phi \]
vanishes. In practice, we minimize the energy $E_u :=\lvert \omega\rvert^2$, where the metric for quaternion-valued $1$-form is defined by
\begin{equation}
\label{eqn:quat_oneform_metric}
\langle \omega ,\eta \rangle := \int_M  \overline{\omega} \wedge (* \eta).
\end{equation}

In discrete case, minimizing the energy $E_u$ again amounts to solving a generalized eigenvalue problem for a $4\lvert F\rvert \times 4\lvert F\rvert$ matrix (see Section \ref{subsec:area_energy}). To avoid introducing the scaling factor as one more function in our representation and subsequently increasing the data size, we first apply the isotropic remeshing with approximate equalized face area \cite{fuhrmann2010direct} for all shapes. In this case the logarithmic factor $u$ should be set to $u_i = \log(1/\sqrt{\lvert \tilde{A}_i\rvert})$, where $\tilde{A}_i$ is the face area of the conformal parameterization. 

\begin{figure}[ht!]
\centering
\def\svgwidth{0.9\columnwidth}
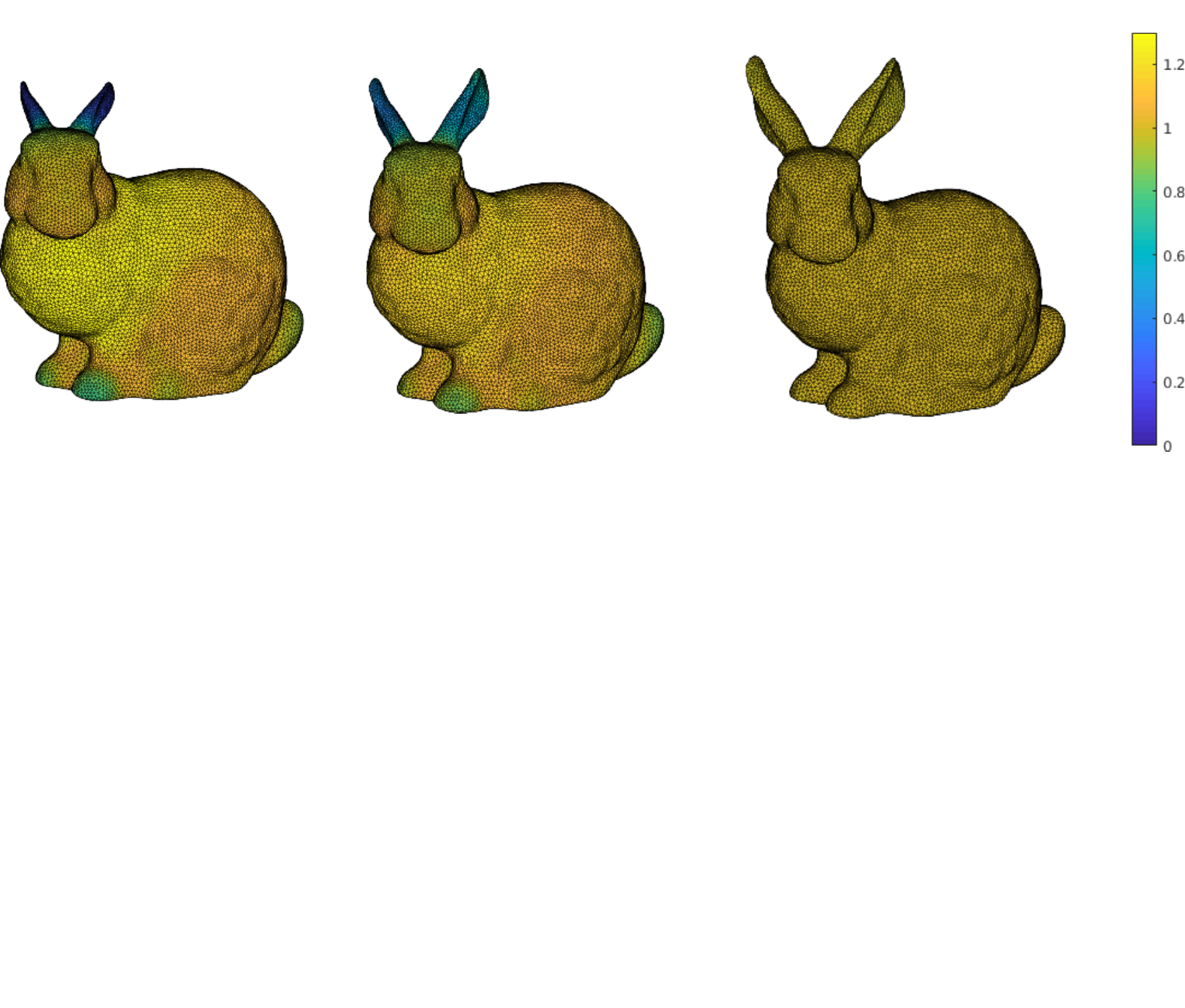
\label{fig:area_correction}\caption{Reconstruction of shapes from their conformal parameterization. While the Willmore energy is defined by $W = \sum_i h_i^2$, we define the relative Willmore energy between two meshes with identical connectivity by $r.W :=\sum_i ((h_1)_i-(h_2)_i)^2$, which measures how close the mean curvature half-density of two meshes are. This experiment shows that our method substantially improves the accuracy of curvature reconstruction. Furthermore, the area distortion, which usually appears in the regions with high curvature, gets much reduced by the area calibration. Note that, in contrast to \cite{chern2015_close-to-conformal-deformations-of-volumes}, we only encode the expected scaling factor in the energy $\lvert \omega\rvert^2$ and the factual scaling factor $\lvert \phi\rvert^4$ is determined by the optimizer. }

\end{figure}

In summary, we first minimize the energy $E_D$ with a small step length several times until the mean curvature half-density converges to the prescribed one. Then, we minimize the energy $E_u$ once to get the correct area scaling factor.

\begin{figure}[h!]
\centering
\def\svgwidth{1\columnwidth}
\begingroup%
  \makeatletter%
  \providecommand\color[2][]{%
    \errmessage{(Inkscape) Color is used for the text in Inkscape, but the package 'color.sty' is not loaded}%
    \renewcommand\color[2][]{}%
  }%
  \providecommand\transparent[1]{%
    \errmessage{(Inkscape) Transparency is used (non-zero) for the text in Inkscape, but the package 'transparent.sty' is not loaded}%
    \renewcommand\transparent[1]{}%
  }%
  \providecommand\rotatebox[2]{#2}%
  \newcommand*\fsize{\dimexpr\f@size pt\relax}%
  \newcommand*\lineheight[1]{\fontsize{\fsize}{#1\fsize}\selectfont}%
  \ifx\svgwidth\undefined%
    \setlength{\unitlength}{675bp}%
    \ifx\svgscale\undefined%
      \relax%
    \else%
      \setlength{\unitlength}{\unitlength * \real{\svgscale}}%
    \fi%
  \else%
    \setlength{\unitlength}{\svgwidth}%
  \fi%
  \global\let\svgwidth\undefined%
  \global\let\svgscale\undefined%
  \makeatother%
  \begin{picture}(1,0.25714289)%
    \lineheight{1}%
    \setlength\tabcolsep{0pt}%
    \put(0,0){\includegraphics[width=\unitlength,page=1]{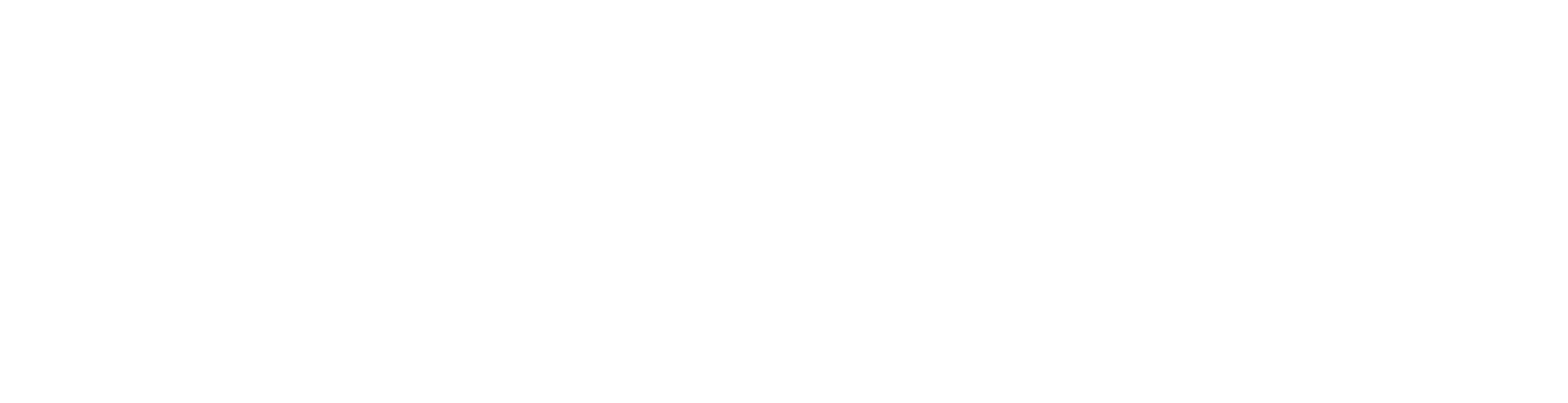}}%
    \put(0.30906853,0.00346891){\color[rgb]{0,0,0}\makebox(0,0)[lt]{\lineheight{1.25}\smash{\begin{tabular}[t]{l}$\lvert V\rvert=3738$\end{tabular}}}}%
    \put(0.06065587,0.00267526){\color[rgb]{0,0,0}\makebox(0,0)[lt]{\lineheight{1.25}\smash{\begin{tabular}[t]{l}$\lvert V\rvert=2442$\end{tabular}}}}%
    \put(0.52335428,0.00267526){\color[rgb]{0,0,0}\makebox(0,0)[lt]{\lineheight{1.25}\smash{\begin{tabular}[t]{l}$\lvert V\rvert=5018$\end{tabular}}}}%
    \put(0.75430711,0.00346891){\color[rgb]{0,0,0}\makebox(0,0)[lt]{\lineheight{1.25}\smash{\begin{tabular}[t]{l}$\lvert V\rvert=10001$\end{tabular}}}}%
    \put(0,0){\includegraphics[width=\unitlength,page=2]{figs/remeshing/remeshing.pdf}}%
  \end{picture}%
\endgroup%

\caption{Remeshing. Given an original shape with $\lvert V\rvert = 5000$, the density is modified by multiplying with $0.25$, $0.75$, $1$ and $2$. The mean curvature half-density changes accordingly such that the mean curvature is preserved. }\label{fig:remeshing}
\end{figure}

\begin{figure*}
\centering
\def\svgwidth{0.95\textwidth}
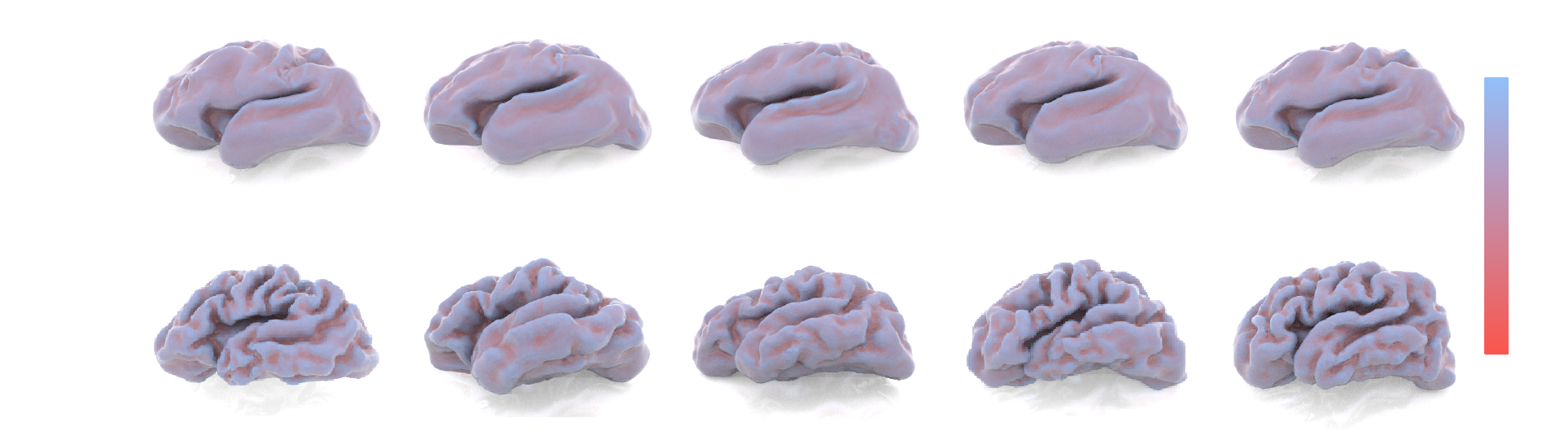
\caption{The randomly generated cortical surfaces by Multi-chart GAN \cite{Ben_Hamu_2018}
and the VAE based on our representation. Our representation has dimension $320 \times 32\times 32\times 2 = 655360$, which has the same magnitude as the data size of Multi-chart, i.e., $16\times 64\times 64\times 3 = 196608$. However, we only require 3 landmark points for alignment, while the Multi-chart needs a dense correspondence for surface reconstruction. The surfaces are labeled by the mean curvature half-density. Note that, the training data mostly have the Willmore energy from 900 to 1000. Although the generated surfaces from our model have been smoothed to a certain extent (partly due to a well-known limitation of VAE), our model apparently preserves more fine structures than the position-based model.}\label{fig:brain2}
\end{figure*}

\section{Results}

We use the Matlab package gptoolbox \cite{gptoolbox} for data pre-processing and Tensorflow \cite{tensorflow2015-whitepaper} for building the neural networks on meshes. All the neural networks are trained and evaluated with the GPU GeForce GTX 1080 with 8GB memory.

\subsection{Preliminary applications}

\label{subsec:pre}
We first present some simple applications that are unrelated to machine learning. 

In smooth case, the mean curvature half-density changes covariantly $h\mapsto m\cdot h$ under the parameterization scaling $x\mapsto m\cdot x$, $m\in \mathbb{R}$. Analogously, in discrete case, one can adjust the parameterization by scaling the vertex density, i.e.,  multiplies the density $\mathfrak{d}$ with a constant number, $\mathfrak{d} \mapsto m \mathfrak{d}$. In order to preserve the shape, one has to adjust the mean curvature half-density by $h\mapsto \frac{h}{\sqrt{m}}$. The shapes reconstructed from the modified representation are actually remeshings with approximately $m\lvert V\rvert$ vertices, where $\lvert V\rvert$ is the number of vertices of the original mesh. Figure \ref{fig:remeshing} shows that our method will preserve the smooth features on the shape. However, the regions of high curvature tend to be smoothed with declining vertex number.

\paragraph*{Shape interpolation}
We visualize the interpolation of our curvature-based representation. Figure \ref{fig:intp} shows the shapes reconstructed from a linear interpolation of two animals, whose conformal parameterizations are matched by a M\"obius transformation that aligns 3 chosen landmark points. In addition, one can interpolate the latent space representation of a trained autoencoder (see Section \ref{subsec:closed_surf_gen}). Figure \ref{fig:car_intp} shows two latent space bi-linear interpolations of cars. 

\begin{figure}[h!]
\centering
\includegraphics[scale=0.7]{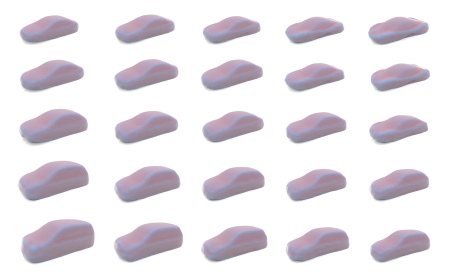}
\caption{Latent space interpolation. We choose four examples in the car dataset and interpolate their mean values in the latent space of VAE. The left lower triangle is a bilinear interpolation of a van, a car and an SUV. The right upper triangle is a bilinear interpolation of a van, a car and a race car.}\label{fig:car_intp}
\end{figure}

\paragraph*{Random generation of disk-like and spherical shapes}
\begin{figure}[h!]
\centering
\def\svgwidth{1\columnwidth}
\begingroup%
  \makeatletter%
  \providecommand\color[2][]{%
    \errmessage{(Inkscape) Color is used for the text in Inkscape, but the package 'color.sty' is not loaded}%
    \renewcommand\color[2][]{}%
  }%
  \providecommand\transparent[1]{%
    \errmessage{(Inkscape) Transparency is used (non-zero) for the text in Inkscape, but the package 'transparent.sty' is not loaded}%
    \renewcommand\transparent[1]{}%
  }%
  \providecommand\rotatebox[2]{#2}%
  \newcommand*\fsize{\dimexpr\f@size pt\relax}%
  \newcommand*\lineheight[1]{\fontsize{\fsize}{#1\fsize}\selectfont}%
  \ifx\svgwidth\undefined%
    \setlength{\unitlength}{571.35327737bp}%
    \ifx\svgscale\undefined%
      \relax%
    \else%
      \setlength{\unitlength}{\unitlength * \real{\svgscale}}%
    \fi%
  \else%
    \setlength{\unitlength}{\svgwidth}%
  \fi%
  \global\let\svgwidth\undefined%
  \global\let\svgscale\undefined%
  \makeatother%
  \begin{picture}(1,0.38645621)%
    \lineheight{1}%
    \setlength\tabcolsep{0pt}%
    \put(0,0){\includegraphics[width=\unitlength,page=1]{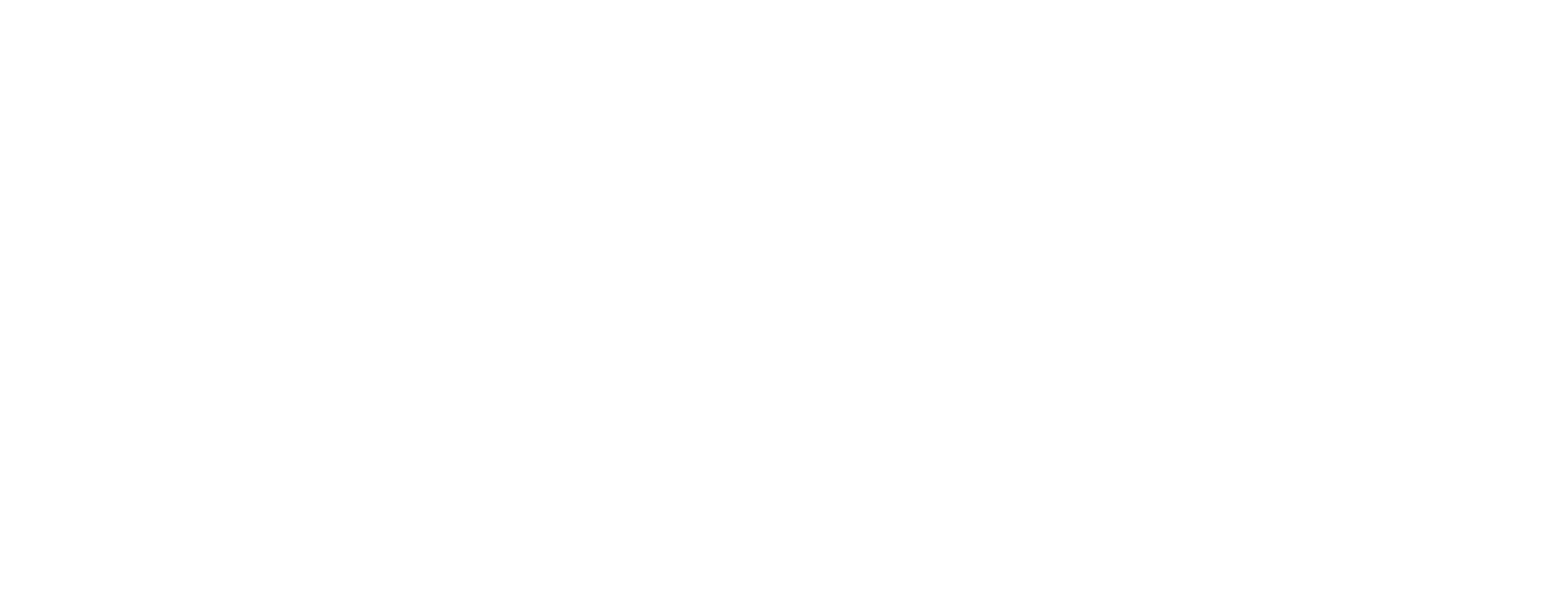}}%
    \put(0.13663622,0.00728123){\color[rgb]{0,0,0}\makebox(0,0)[lt]{\lineheight{1.25}\smash{\begin{tabular}[t]{l}any conformal map\end{tabular}}}}%
    \put(0.59869755,0.00765627){\color[rgb]{0,0,0}\makebox(0,0)[lt]{\lineheight{1.25}\smash{\begin{tabular}[t]{l}alignment\end{tabular}}}}%
    \put(0.48805758,0.34895126){\color[rgb]{0,0,0}\makebox(0,0)[lt]{\lineheight{1.25}\smash{\begin{tabular}[t]{l}$f$\end{tabular}}}}%
    \put(0,0){\includegraphics[width=\unitlength,page=2]{figs/alignment/alignment.pdf}}%
  \end{picture}%
\endgroup%

\caption{For disk-like surfaces, given two landmark points there is a unique conformal map which maps the first point (red) to zero and maps the second one (blue) to the $x$-axis.}\label{fig:alignment}
\end{figure}
We test our model for disk-like surfaces on a dataset of anatomical shapes provided by  \cite{Boyer_2011}. In particular, we choose the shapes of teeth, which is one of three types of bone in this dataset. To create the representation, we first take an intermediate conformal map, which maps the teeth to the unit disk by the algorithm from \cite{Choi_2015}.

Several landmark points are available in  \cite{Boyer_2011}, hence we choose two landmark points $u_i$, $v_i$  for every shape $M_i$. We know that the conformal automorphisms of the unit disk have the form
\[f(z) = e^{i\theta} \frac{z-a}{1-\overline{a}z},\]
where $\theta\in\mathbb{R}$ and $a \in \mathbb{C}$. Set $a= u_i$ and $\theta$ such that $f(v_i) \in \mathbb{R}$. Clearly, this uniquely determined map $f_{a,\theta}$ satisfies $f(u_i) = 0 $ and $f(v_i)\in \mathbb{R}$. Fixing a reference shape $M_0$, for any shape $M_i$ we apply the alignment map $f_0^{-1}\circ f_i$ for every shapes. 

All the aligned disk meshes are then mapped to the square via the Schwarz-Christoffel mapping. The functions are interpolated on the $256 \times 256$ grid using the \texttt{scatteredInterpolant} function in Matlab. 

For spherical surfaces we take the dataset of 1240 cars from ShapeNet \cite{shapenet2015}. All the shapes are converted into genus-$0$ surfaces by Umetani \cite{Umetani_2017}. Then we create the aligned conformal parameterization by the canonical M\"obius transformation \cite{doi:10.1111/cgf.13503}. The canonical domain with is obtained by subdividing the icosahedron twice so it has $20\times 4^2 = 320$ faces. Each face is assigned with a $32 \times 32$ grid. Hence, each shape is represented by a $320 \times 32\times 32\times 2$-dimensional tensors.

The randomly generated teeth and cars are shown in the appendix as well as their curvature representation.

\subsection{Generation of unaligned data}
\label{subsec:unaligned}

\paragraph*{Discussion of local invariance}
\begin{wrapfigure}{r}{0.3\columnwidth}

\vspace{-\intextsep}
\hspace*{-1.2\columnsep}
\def\svgwidth{0.4\columnwidth}
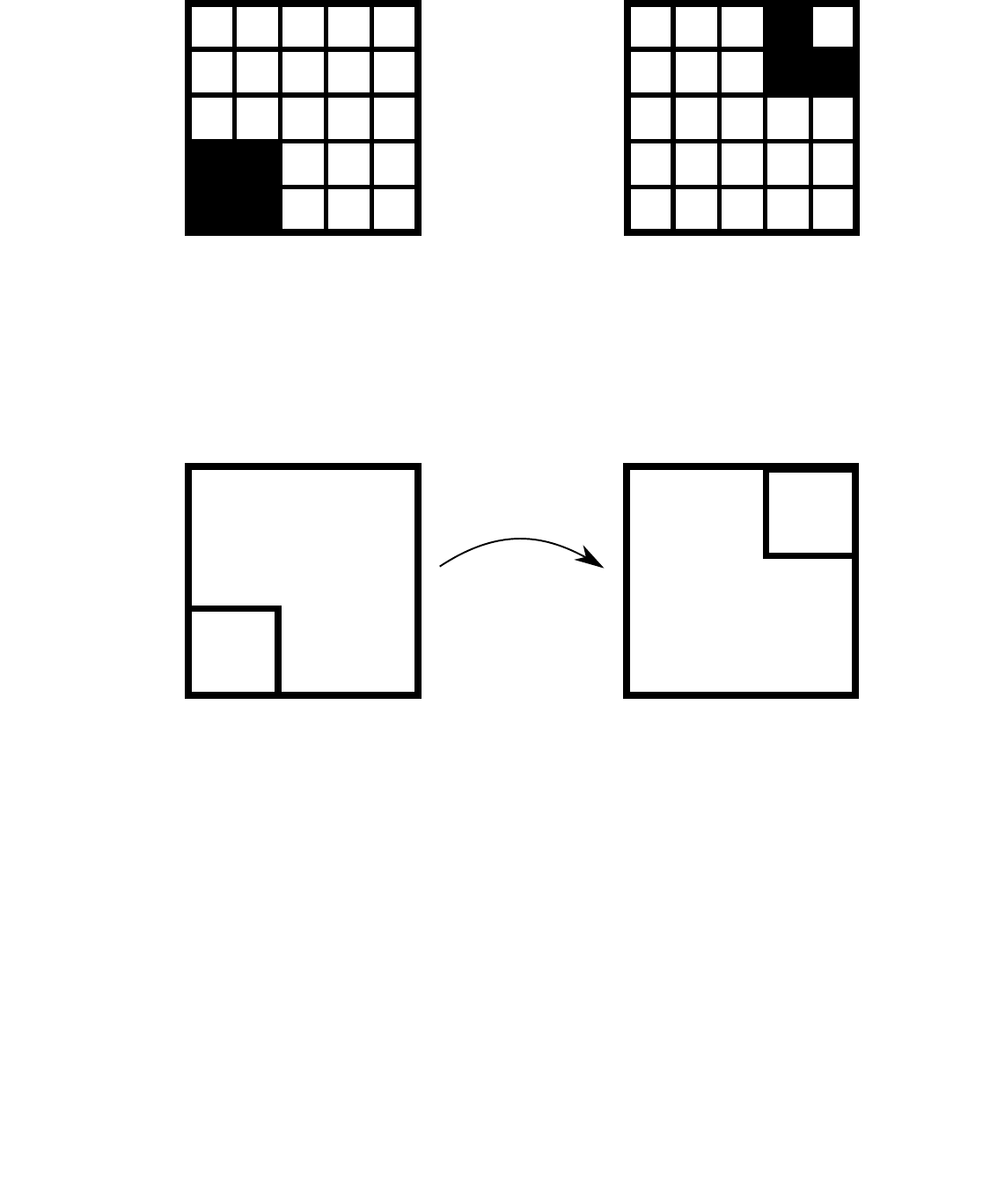\label{fig:local_inv}
\end{wrapfigure}
We call two functions $f_1$ and $f_2$ local invariant if they have the same function value but only differ by a transformation $g$ of domain, i.e., $f_1 = f_2 \circ g$. Traditional CNNs are able to capture the translational features such as (a) of inset. Hence one would expect the CNNs for 3D shapes with the similar properties like local invariance under translation, rotation or even scaling. However, 3D generative models based on position, such as point cloud and mesh, will not have such properties due to the varied function value of coordinates (see (b)). This makes it more difficult for CNNs to extract meaningful information. The voxel-based models are local invariant, but they are not applicable for data with high resolution due to the high cost of memory and computation. Some multi-resolution representations, e.g., octree \cite{ogn2017,Wang_2018}, are designed to overcome this problem, but the local invariant property does not hold anymore. In contrast, our model (sketched by (c)), together with the CNN on the sphere, provides an efficient way to learn the 3D data without a certain alignment. We verify our argument with the following two examples.

\paragraph*{Learning unaligned anatomical data}

\begin{wrapfigure}{r}{0.3\columnwidth}

\vspace{-\intextsep}
\hspace*{-1.1\columnsep}
\includegraphics[width = 0.4\columnwidth]{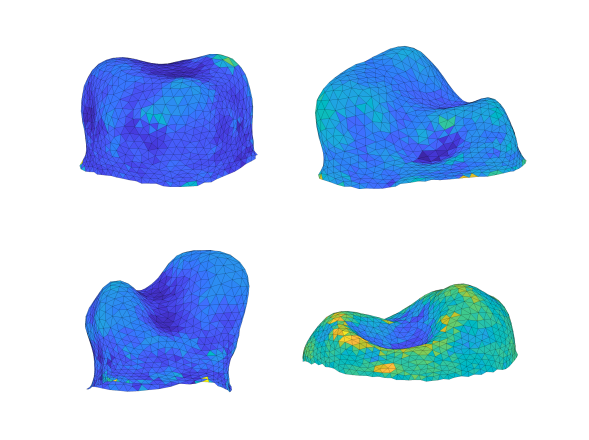}
\vspace{-2\intextsep}
\end{wrapfigure}
We merge three different anatomical models in \cite{Boyer_2011} and create the representations without any alignment methods. Insect shows the randomly generated bones of different types. Compared with Figure \ref{fig:car} the bones get smoothed due to the expanded shape space. However, we show that our model is still capable to extract the meaningful information from the ambiguity by visualizing the latent space distribution (Figure \ref{fig:cluster}). We compare the result with a baseline model that has the same network architecture but operates on the coordinate functions.

\begin{figure}[h!]
\centering

\def\svgwidth{1\columnwidth}
\begingroup%
  \makeatletter%
  \providecommand\color[2][]{%
    \errmessage{(Inkscape) Color is used for the text in Inkscape, but the package 'color.sty' is not loaded}%
    \renewcommand\color[2][]{}%
  }%
  \providecommand\transparent[1]{%
    \errmessage{(Inkscape) Transparency is used (non-zero) for the text in Inkscape, but the package 'transparent.sty' is not loaded}%
    \renewcommand\transparent[1]{}%
  }%
  \providecommand\rotatebox[2]{#2}%
  \newcommand*\fsize{\dimexpr\f@size pt\relax}%
  \newcommand*\lineheight[1]{\fontsize{\fsize}{#1\fsize}\selectfont}%
  \ifx\svgwidth\undefined%
    \setlength{\unitlength}{542.09530688bp}%
    \ifx\svgscale\undefined%
      \relax%
    \else%
      \setlength{\unitlength}{\unitlength * \real{\svgscale}}%
    \fi%
  \else%
    \setlength{\unitlength}{\svgwidth}%
  \fi%
  \global\let\svgwidth\undefined%
  \global\let\svgscale\undefined%
  \makeatother%
  \begin{picture}(1,0.33704966)%
    \lineheight{1}%
    \setlength\tabcolsep{0pt}%
    \put(0,0){\includegraphics[width=\unitlength,page=1]{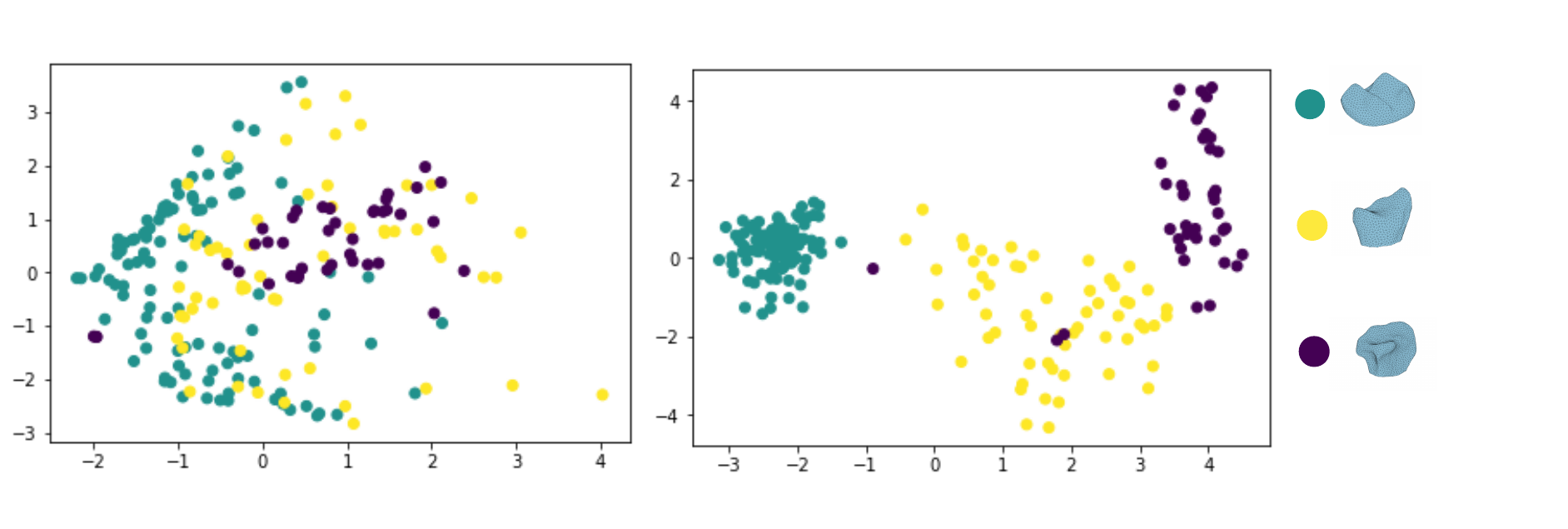}}%
    \put(0.1086354,0.00263302){\color[rgb]{0,0,0}\makebox(0,0)[lt]{\lineheight{1.25}\smash{\begin{tabular}[t]{l}coordinate\end{tabular}}}}%
    \put(0.54187507,0.00026121){\color[rgb]{0,0,0}\makebox(0,0)[lt]{\lineheight{1.25}\smash{\begin{tabular}[t]{l}curvature\end{tabular}}}}%
    \put(0.82623261,0.23440714){\color[rgb]{0,0,0}\makebox(0,0)[lt]{\lineheight{1.25}\smash{\begin{tabular}[t]{l}\small teeth\end{tabular}}}}%
    \put(0.83312527,0.15470072){\color[rgb]{0,0,0}\makebox(0,0)[lt]{\lineheight{1.25}\smash{\begin{tabular}[t]{l}\small mt1\end{tabular}}}}%
    \put(0.82443457,0.06704256){\color[rgb]{0,0,0}\makebox(0,0)[lt]{\lineheight{1.25}\smash{\begin{tabular}[t]{l}\small radius\end{tabular}}}}%
  \end{picture}%
\endgroup%

\caption{Latent space visualization. The dataset is composed of three different types of anatomical surfaces. We project the latent space representation on a $2$-dimensional space by PCA. Though all the shapes are packed without alignment, the three types of bones are clearly separated in the latent space. In contrast, the model based on the coordinate failed to learn the structure of the bones, so their distribution in the latent space is not well separated.} \label{fig:cluster}
\end{figure}

\paragraph*{Generation of transformed cars}
In this experiment we would like to see whether the 3D generative models are able to correctly predict shapes with various transformations. The dataset is created by randomly translating, rotating and scaling a single shape of car in the cube of size $[-1,1]\times [-1,1]\times [-1,1]$. We train autoencoders based on different models on $900$ training data and test them on $100$ validation data. The comparison shows that our method produces more accurate predictions than others (see Figure \ref{fig:trcar}). Since only our model considers the mesh structure of shapes, to make a fair comparison, we evaluate the results with Chamfer distance which only depends on the underlying point clouds. Note that, as a trade-off, our representation loses the information of translation and scaling. Thus we first normalize the shapes reconstructed from our model and then calculate the Chamfer distance to the ground truth.

\subsection{Cortical surface generation}
\label{subsec:closed_surf_gen}
\toremove{
\paragraph*{Random generation of cars}
The representation is built over the aligned spherical metrics made by M\"obius registration \cite{doi:10.1111/cgf.13503}. The function of mean curvature half-density is linearly scaled in the range $[-0.2,0.5]$ and the logarithmic density is scaled in $[5,10]$.
The architecture of VAE is depicted in Fig.~\ref{fig:vae}. We train the models with batch size $20$ and $200$ epochs. The randomly generated cars are shown in Fig.~\ref{fig:car}. The shapes of car are reconstructed by minimizing the Dirac energy with step length $1$ once without area correction.
}

To show that our model is particularly good at preserving the fine structure, we perform the experiment on human cortical surfaces, which are highly folded with a lot of "hills" and "valleys". A dataset of cortical surfaces are available on the Open Access Series of Imaging Studies (OASIS) \cite{marcus2010open}. The MRI images are converted to genus-$0$ surfaces via the open-source reconstruction software FreeSurfer (http://surfer.nmr.mgh.harvard.edu/).  

We first compare our model to three other state-of-art autoencoders for 3D shapes. Figure \ref{fig:brain1} shows that, although all models succeed in characterizing the shapes in a large scale, our model preserves much more small features, e.g., the curvature, than the others.

\paragraph*{Training details} Our model and the baseline model are trained with 200 epochs for around 5 hours. The point-cloud AE \cite{Achlioptas2018LearningRA} with 2048 points for each data and AtlasNet \cite{groueix2018} with 2500 points for each data are both trained with 500 epochs for  approximately 4 hours. Although the point-cloud based models above have smaller data size than ours, the training of their neural networks already exhausted our GPU memory. The OGN, with the octree representation of $128\times 128 \times 128$-dimensional voxels, is trained with 4000 epochs with 5 hours. While other models produce the shapes instantly after training, it takes 2 minutes with our method to reconstruct a mesh with $10000$ vertices from curvature.

Next, we compare the cortical surfaces 
randomly generated by our VAE to the ones by Multi-chart GAN \cite{Ben_Hamu_2018} (Figure \ref{fig:brain2}). While both mesh-based models generate significantly more faithful results than other types of representation in Figure \ref{fig:brain1}, the "hills" and "valleys" are much more visible with our model. Moreover, we only choose 3 landmark points on each shape to align the conformal parameterization, while it requires $21$ landmark points to create $16$ charts as in \cite{Ben_Hamu_2018}, and even a template shape, which amounts to a dense correspondence, to reconstruct the final shapes.

At last, we try to create an autoencoder that converts the 3D MRI images of brain to cortical surfaces. In this case, the encoder consists of several 3D convolutional layers (see Figure \ref{fig:vol_network}) and the decoder is the same as the ones in previous experiments. Figure \ref{fig:vol2surf} shows that our model is able to predict the cortical surface from the MRI volume to a certain extent, but the accuracy is not yet optimal, because the neural network failed to capture the spatial correspondence between the volumetric data and the spherical data. We leave the construction of a finer 3D-to-2D autoencoder to future work.

\begin{figure*}
\centering
\def\svgwidth{0.9\textwidth}
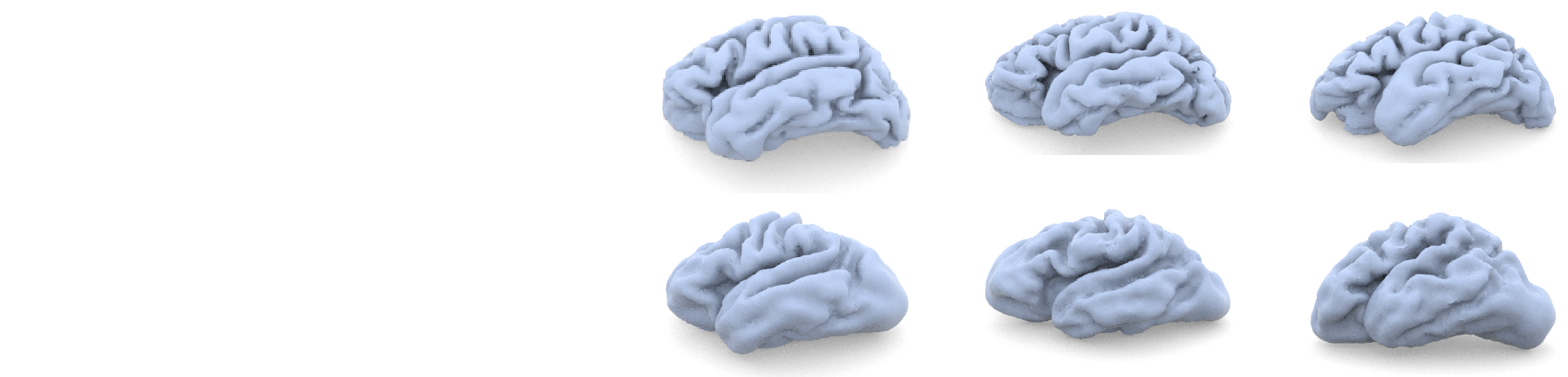
\caption{Volume-Curvature autoencoder. The input is the MRI volumetric data from \cite{marcus2010open}. Since only the left hemisphere is generated, we align the volume using FreeSurfer \cite{reuter2010highly} and chop the volume properly such that the dimension is $193\times 80 \times 195$. The encoder is shown in Figure \ref{fig:vol_network} and the decoder is the same as Figure \ref{fig:sph_network}.}\label{fig:vol2surf}
\end{figure*}

\section{Limitations and future work}
\begin{wrapfigure}{r}{0.3\columnwidth}
\centering
\hspace*{-1\columnsep}
\def\svgwidth{0.4\columnwidth}
%
\begingroup%
  \makeatletter%
  \providecommand\color[2][]{%
    \errmessage{(Inkscape) Color is used for the text in Inkscape, but the package 'color.sty' is not loaded}%
    \renewcommand\color[2][]{}%
  }%
  \providecommand\transparent[1]{%
    \errmessage{(Inkscape) Transparency is used (non-zero) for the text in Inkscape, but the package 'transparent.sty' is not loaded}%
    \renewcommand\transparent[1]{}%
  }%
  \providecommand\rotatebox[2]{#2}%
  \newcommand*\fsize{\dimexpr\f@size pt\relax}%
  \newcommand*\lineheight[1]{\fontsize{\fsize}{#1\fsize}\selectfont}%
  \ifx\svgwidth\undefined%
    \setlength{\unitlength}{475.50000288bp}%
    \ifx\svgscale\undefined%
      \relax%
    \else%
      \setlength{\unitlength}{\unitlength * \real{\svgscale}}%
    \fi%
  \else%
    \setlength{\unitlength}{\svgwidth}%
  \fi%
  \global\let\svgwidth\undefined%
  \global\let\svgscale\undefined%
  \makeatother%
  \begin{picture}(1,0.88664081)%
    \lineheight{1}%
    \setlength\tabcolsep{0pt}%
    \put(0,0){\includegraphics[width=\unitlength,page=1]{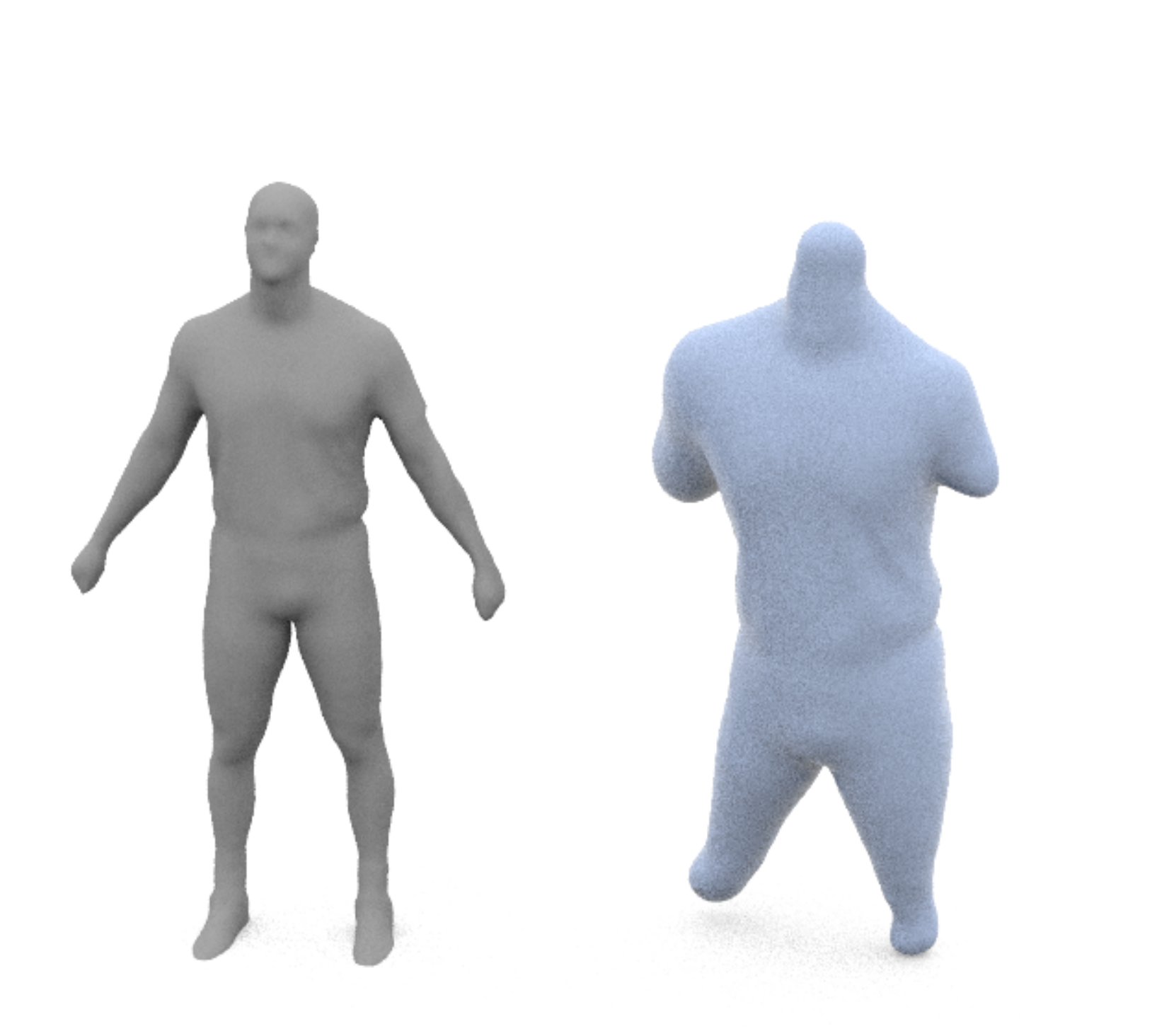}}%
    \put(0.04889582,0.8690853){\color[rgb]{0,0,0}\makebox(0,0)[lt]{\lineheight{1.25}\smash{\begin{tabular}[t]{l}Ground \end{tabular}}}}%
    \put(0.51419556,0.87066256){\color[rgb]{0,0,0}\makebox(0,0)[lt]{\lineheight{1.25}\smash{\begin{tabular}[t]{l}Failed\end{tabular}}}}%
    \put(0.08066691,0.77422255){\color[rgb]{0,0,0}\makebox(0,0)[lt]{\lineheight{1.25}\smash{\begin{tabular}[t]{l}Truth\end{tabular}}}}%
    \put(0.48602973,0.77895442){\color[rgb]{0,0,0}\makebox(0,0)[lt]{\lineheight{1.25}\smash{\begin{tabular}[t]{l}Example\end{tabular}}}}%
  \end{picture}%
\endgroup%

\end{wrapfigure}
First, currently it is difficult to model the shapes like long tubes, such as arms and legs of human, because the conformal parameterization of such shapes always has extremely large area distortion. The information easily gets lost while being transferred from such regions to the canonical domain (inset), unless one uses a domain with extremely high resolution. A solution might be a multi-resolution data structure, such as \cite{Grinspun_2002, Wang_2018}. Then it is desirable to design a structure of neural network that is specifically adapted to such multi-resolutional data structures.

Second, to make our model fully rotational invariant rather than just local invariant, one might combine our representation with the equivariant neural networks by Cohen et al. \cite{s.2018spherical}, so that the alignment procedure can be completely removed. Then it would be interesting to develop a corresponding decoder network.

\section{Conclusion}
We propose a novel intrinsic representation of 3D surfaces based on mean curvature and  metric. A 3D generative model is built based on this representation and it manifests better performance than other models in capturing the fine structure and the symmetry of the ambient space.

\bibliographystyle{eg-alpha-doi} 
\bibliography{reference}       

\section{Appendix}
\subsection{Compute the weighted centroid of polygons}
\label{app:CVT}
The weighted centroid of a polygon is given by
\begin{equation*}
 v^* = \frac{\int_ V y \mathfrak{d}(y) dy}{\int_V  \mathfrak{d}(y) dy}.
\end{equation*}
A Voronoi cell is naturally decomposed in several triangles, of which we first compute the weighted centroid.

\begin{wrapfigure}{r}{0.25\columnwidth}
\centering
\hspace*{-1\columnsep}
\def\svgwidth{0.25\columnwidth}
\begingroup%
  \makeatletter%
  \providecommand\color[2][]{%
    \errmessage{(Inkscape) Color is used for the text in Inkscape, but the package 'color.sty' is not loaded}%
    \renewcommand\color[2][]{}%
  }%
  \providecommand\transparent[1]{%
    \errmessage{(Inkscape) Transparency is used (non-zero) for the text in Inkscape, but the package 'transparent.sty' is not loaded}%
    \renewcommand\transparent[1]{}%
  }%
  \providecommand\rotatebox[2]{#2}%
  \newcommand*\fsize{\dimexpr\f@size pt\relax}%
  \newcommand*\lineheight[1]{\fontsize{\fsize}{#1\fsize}\selectfont}%
  \ifx\svgwidth\undefined%
    \setlength{\unitlength}{222.13217242bp}%
    \ifx\svgscale\undefined%
      \relax%
    \else%
      \setlength{\unitlength}{\unitlength * \real{\svgscale}}%
    \fi%
  \else%
    \setlength{\unitlength}{\svgwidth}%
  \fi%
  \global\let\svgwidth\undefined%
  \global\let\svgscale\undefined%
  \makeatother%
  \begin{picture}(1,0.96901728)%
    \lineheight{1}%
    \setlength\tabcolsep{0pt}%
    \put(0,0){\includegraphics[width=\unitlength,page=1]{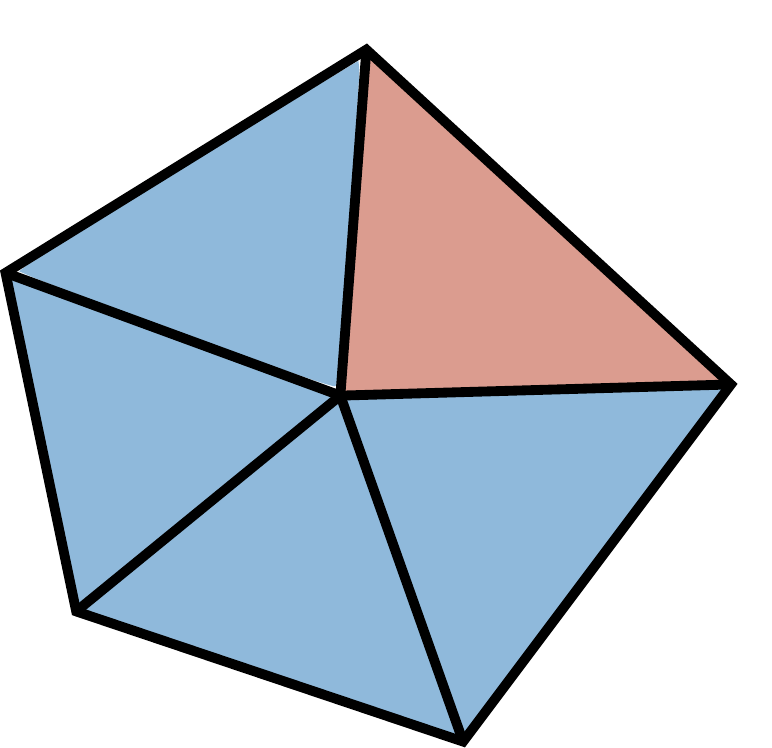}}%
    \put(0.37277216,0.52554996){\color[rgb]{0,0,0}\makebox(0,0)[lt]{\lineheight{1.25}\smash{\begin{tabular}[t]{l}$v_1$\end{tabular}}}}%
    \put(0.8686162,0.48623935){\color[rgb]{0,0,0}\makebox(0,0)[lt]{\lineheight{1.25}\smash{\begin{tabular}[t]{l}$v_2$\end{tabular}}}}%
    \put(0.36626058,0.93481397){\color[rgb]{0,0,0}\makebox(0,0)[lt]{\lineheight{1.25}\smash{\begin{tabular}[t]{l}$v_3$\end{tabular}}}}%
  \end{picture}%
\endgroup%

\end{wrapfigure}
Denote the density on the vertex $v_i$ by $\mathfrak{d}_i$ and we assume that the density is linearly interpolated on every triangle. The denominator of $v^*$ is called the weighted area, which is given by $\mathfrak{A}_i = \frac{\mathfrak{d}(v_1) + \mathfrak{d}(v_2) + \mathfrak{d}(v_3)}{3} A_i$, where $A_i$ is the triangle area.

Integrating the linear function on the triangle $i$, we obtain
\[v_i^* = \frac{(2\mathfrak{d}_1+ \mathfrak{d}_2+ \mathfrak{d}_3)v_1 + (\mathfrak{d}_1 + 2\mathfrak{d}_2 + \mathfrak{d}_3)v_2 + (\mathfrak{d}_1 + \mathfrak{d}_2 + 2\mathfrak{d}_3)v_3}{4(\mathfrak{d}_1 + \mathfrak{d}_2+ \mathfrak{d}_3)}.\]

Then, the centroid of the polygon is the weighted sum
\[v^* =  \frac{\sum_i v_i^* \cdot \mathfrak{A}_i}{\sum_i \mathfrak{A}_i}.\]

\subsection{Proof of Thm. \ref{thm:closing} (Closing condition for prescribing the area factor)}
\label{app:closing}
\begin{proof}
Let $(x,y)$ be a conformal coordinate of the immersion $f:M\rightarrow\mathbb{R}^3$. The left-hand side of \eqref{eqn:closing_area} is actually 
\[\phi_x \cdot \phi^{-1} dx + \phi_y \cdot \phi^{-1} dy, \]
while the right hand side reads
\begin{align}
-\frac{1}{2} &(-u_x f_x^{-1} -u_y f_y^{-1})\cdot (f_x dx + f_y dy) \nonumber\\
&= \frac{1}{2}((u_x + u_y f_y^{-1}f_x) dx  + (u_y+u_xf_x^{-1}f_y)dy) \nonumber\\
&= \frac{1}{2}( (u_x + u_y n) dx + (u_y-u_xn) dy). \label{eqn:rhs}
\end{align}
The equation \eqref{eqn:closing_area} implies that
\begin{align*}
\phi_x \cdot \phi^{-1} &= \frac{1}{2}(u_x+u_yn),\\
\phi_y \cdot \phi^{-1} &= \frac{1}{2}(u_y-u_xn).
\end{align*}
Substituting the equations above into the Dirac operator \eqref{eqn:dirac_eq} in local form, we obtain
\begin{align*}
D_f \phi &= f_x \phi_y - f_y \phi_x\\
&= \frac{1}{2}f_x(u_y-u_xn) - \frac{1}{2}f_y(u_x+u_yn)\\
&=0, 
\end{align*}
by $f_x\cdot n= -f_y$ and $f_y\cdot n = -f_x$.
\end{proof}
\toremove{It suffices to check the terms with $dx$ on both sides. Clearly, on the right hand side the real part  of $dx$ is $\frac{1}{2}u_x$, which coincides with the real part of $\phi_x\cdot \phi^{-1}$, since $e^u = \lvert \phi\rvert^2 = \phi \overline{\phi} \Rightarrow e^u u_x = \phi_x \overline{\phi} + \phi \overline{\phi_x}$, and hence 
\begin{equation}
\label{eqn:ux}
u_x = \phi_x\phi^{-1} + \overline{\phi}^{-1} \overline{\phi_x} = 2\re(\phi_x \phi^{-1}).
\end{equation}

The Dirac equation \eqref{eqn:dirac_eq} can be written in local form as
\[f_x \phi_y - f_y \phi_x = c \phi,\]
where $c = \lvert df\rvert^2 \rho$ is a real-valued function. It follows that $\phi_y = f_x^{-1}f_y \phi_x + cf_x^{-1} \phi = -n \phi_x - \frac{c f_x \phi}{\lvert f_x\rvert^2}$. By the same argument as \eqref{eqn:ux} we have $u_y=\phi_y\phi^{-1} + \overline{\phi}^{-1} \overline{\phi_y}$. Hence the imaginary part of $dx$ in \eqref{eqn:rhs} is 
\begin{align*}
\frac{1}{2}u_x &=\frac{1}{2}(\phi_y \phi^{-1} + \overline{\phi}^{-1} \overline{\phi_y}) n \\
&= \frac{1}{2}\left(\left( -n \phi_x -  \frac{c f_x \phi}{\lvert f_x\rvert^2}\right)\phi^{-1} + \overline{\phi}^{-1}\left(\overline{\phi_x}n + \frac{c\overline{\phi}f_x}{\lvert f_x\rvert^2}\right)\right)n \\
&= \frac{1}{2}\left( \phi_x \phi^{-1} - \overline{\phi}^{-1} \overline{\phi_x}\right) = \im(\phi_x \phi^{-1}),
\end{align*}
where $-n\phi_xn = \phi_x$ is used.}

\subsection{Finite element method for quaternion gradient}
To obtain the discrete formula of the energy $\lvert \omega^2\rvert$, we first derive the formula of the quaternion gradient in discrete case.

Let $h:M\rightarrow \mathbb{R}$ be any function. We know that the gradient is defined by $\grad u:= (du)^\sharp$, where $\sharp:T^*M \rightarrow TM$ is called raising indices defined by
\[\langle \omega^\sharp ,v\rangle  = \omega (v),\quad \text{ for any } v \in TM\]

\begin{wrapfigure}{r}{0.25\columnwidth}
\centering
\hspace*{-1\columnsep}
\def\svgwidth{0.3\columnwidth}
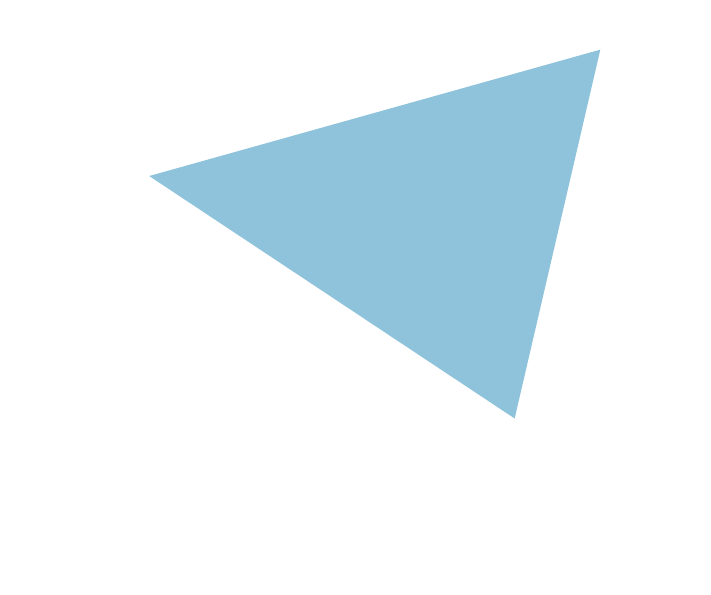
\end{wrapfigure}
In a triangle $i$ in quaternion space with the oriented edges $a$, $b$, $c\in \mathbb{H}$, we choose a coordinate $(x,y)$ system (inset). Assuming that $h$ is a linear function with the value $h_1$, $h_2$, $h_3$ at the vertices, write $dh$ in local form as:
\[ dh = (h_2-h_1) dx + (h_3-h_1) dy\]
Since $\langle (dx)^\sharp, \partial x\rangle = 1$ and $\langle (dx)^{\sharp} , \partial y \rangle =0$,  $df((dx)^{\sharp})$ is perpendicular to $b$ and has the length $\frac{1}{\lvert c\rvert \sin \theta} = \frac{\lvert b\rvert}{2A}$, where $A$ is the area of the triangle. Thus $df(dx^\sharp) = \frac{n\cdot b}{2A}$ and, by the same argument, we have $df(dy^\sharp) = \frac{n\cdot c}{2A}$. 

Therefore,
\[\grad_f h = \frac{n}{2A}( a h_1 + b h_2 + c h_3)\]

\subsection{The energy of quaternion $1$-form}
\label{subsec:area_energy}
We discretize the energy 
\[E_u = \lvert \omega\rvert^2 = \lvert d\phi + \frac{1}{2} Gdf\phi\rvert^2 \]
in the scheme of finite element method. In the local coordinate system above, the metric and its inverse read:
\[ g= \begin{pmatrix} \lvert c\rvert^2 & -\langle c,b\rangle  \\ -\langle c,b\rangle & b^2\end{pmatrix},\quad  g^{-1} = \frac{1}{2A} \begin{pmatrix} b^2 & \langle c,b\rangle  \\ \langle c,b\rangle & \lvert c\rvert^2 \end{pmatrix}.\]
With $\omega= \omega_x dx + \omega_y dy $, \eqref{eqn:quat_oneform_metric} becomes 
\[\int (\lvert \omega_x\rvert^2 \lvert b\rvert ^2 +\langle c,b\rangle (\overline{\omega_x}\omega_y + \overline{\omega_y}\omega_x) + \lvert \omega_y\rvert^2 \lvert c\rvert ^2 ) dx\wedge dy.\]
Now, we work out the formula $ \omega  = d \phi + \frac{1}{2} G df \phi$ in one triangle: 
\begin{align*}
\omega &= \left( (\phi_2 -\phi_1 ) +\frac{1}{2} G\cdot c((1-x-y)\phi_1 + x\phi_2 + y\phi_3) \right) dx \\
&+ \left( (\phi_3 -\phi_1) - \frac{1}{2}G\cdot b((1-x-y)\phi_1 + x\phi_2 + y\phi_3) \right)dy 
\end{align*}
where
\begin{align*}
G\cdot c &= u_1 - u_2 + \frac{n}{2A}(-\langle a,c\rangle u_1 - \langle b,c\rangle u_2 -\lvert c\rvert^2 u_3)\\
G\cdot b &= -u_1 + u_3 + \frac{n}{2A}(-\langle a,b\rangle u_1 - \lvert b\rvert^2 u_2 - \langle c,b\rangle u_3)
\end{align*}
The energy $E_u$ is a $\lvert V\rvert\times \lvert V\rvert$ quaternion-valued matrix. With a tedious calculation the entries related to the triangle are given by
\begin{align*}
\lvert \omega\rvert^2_{11} &= \frac{1}{2}\lvert a\rvert^2 - \frac{1}{6}(\lvert a\rvert^2 u_1 + \langle b, a\rangle u_2 + \langle c,a\rangle u_3)+\frac{1}{6} \lvert G\rvert^2 A^2,\\
\lvert \omega\rvert_{23}^2 &= \frac{1}{2}\langle b,c\rangle + \frac{ \lvert G\rvert^2 A^2}{12} \\
&\phantom{=} +  \frac{1}{12}((4An+\lvert a\rvert^2)u_1 -(a\cdot b)u_2 - (c\cdot a) u_3)
\end{align*}
where
\begin{align*}
\lvert G\rvert^2 &= \frac{1}{4A^2}(a^2 u_3^2 + b^2 u_2^2 + c^2 u_1^2 \\
&\phantom{=} + 2\langle a,b\rangle u_1u_2 + 2\langle b,c\rangle u_2u_3 +2 \langle c,a\rangle u_3u_1).
\end{align*}

\begin{figure*}
	\centering
	\def\svgwidth{0.8\textwidth}
	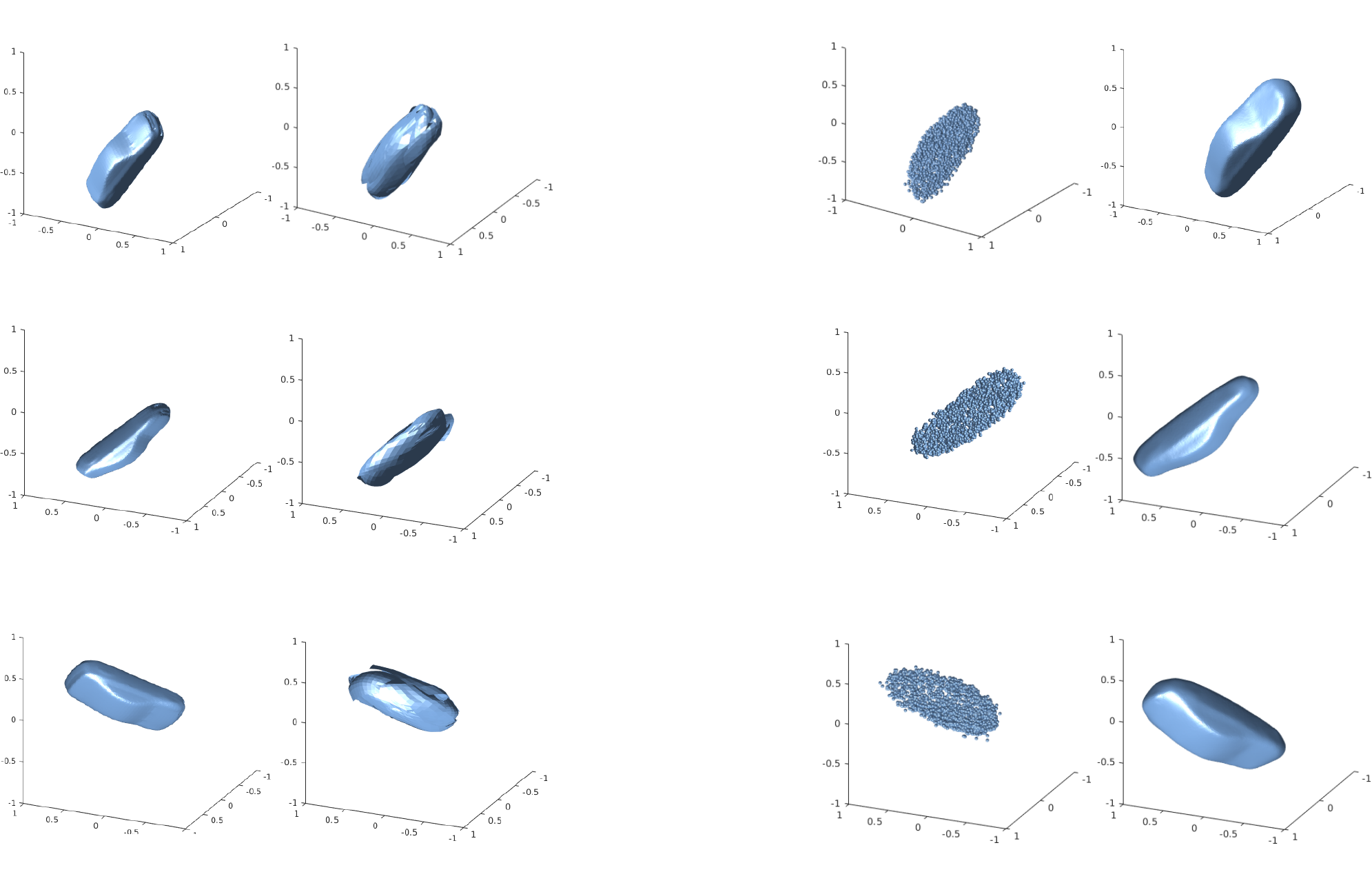
	\caption{Autoencoder for transformed cars. We transform a shape of car by applying random translation, scaling and rotation. We demonstrate our results with other models based on the point clouds, namely the point-cloud AE \cite{Achlioptas2018LearningRA} and the AtlasNet \cite{groueix2018}, the one based on voxels, namely the O-CNN \cite{Wang_2018}. Other methods, though were shown to achieve satisfying results on the aligned dataset, do not correctly capture the symmetry of various transformations. In contrast, our model succeeds in producing convincing transformed shapes. We evaluate the results by measuring the Chamfer distance $\overline{CD}$. However, since our model loses the information of translation and scaling, we have to first normalize the volume of the results with a centered position (unnormalized shapes are shown above). In the end we compute the Chamfer distance of the normalized outputs $\overline{CP}$.}\label{fig:trcar}
\end{figure*}

\begin{figure*}
\centering
\def\svgwidth{0.7\textwidth}
\begingroup%
  \makeatletter%
  \providecommand\color[2][]{%
    \errmessage{(Inkscape) Color is used for the text in Inkscape, but the package 'color.sty' is not loaded}%
    \renewcommand\color[2][]{}%
  }%
  \providecommand\transparent[1]{%
    \errmessage{(Inkscape) Transparency is used (non-zero) for the text in Inkscape, but the package 'transparent.sty' is not loaded}%
    \renewcommand\transparent[1]{}%
  }%
  \providecommand\rotatebox[2]{#2}%
  \newcommand*\fsize{\dimexpr\f@size pt\relax}%
  \newcommand*\lineheight[1]{\fontsize{\fsize}{#1\fsize}\selectfont}%
  \ifx\svgwidth\undefined%
    \setlength{\unitlength}{642.85001843bp}%
    \ifx\svgscale\undefined%
      \relax%
    \else%
      \setlength{\unitlength}{\unitlength * \real{\svgscale}}%
    \fi%
  \else%
    \setlength{\unitlength}{\svgwidth}%
  \fi%
  \global\let\svgwidth\undefined%
  \global\let\svgscale\undefined%
  \makeatother%
  \begin{picture}(1,0.55830232)%
    \lineheight{1}%
    \setlength\tabcolsep{0pt}%
    \put(0,0){\includegraphics[width=\unitlength,page=1]{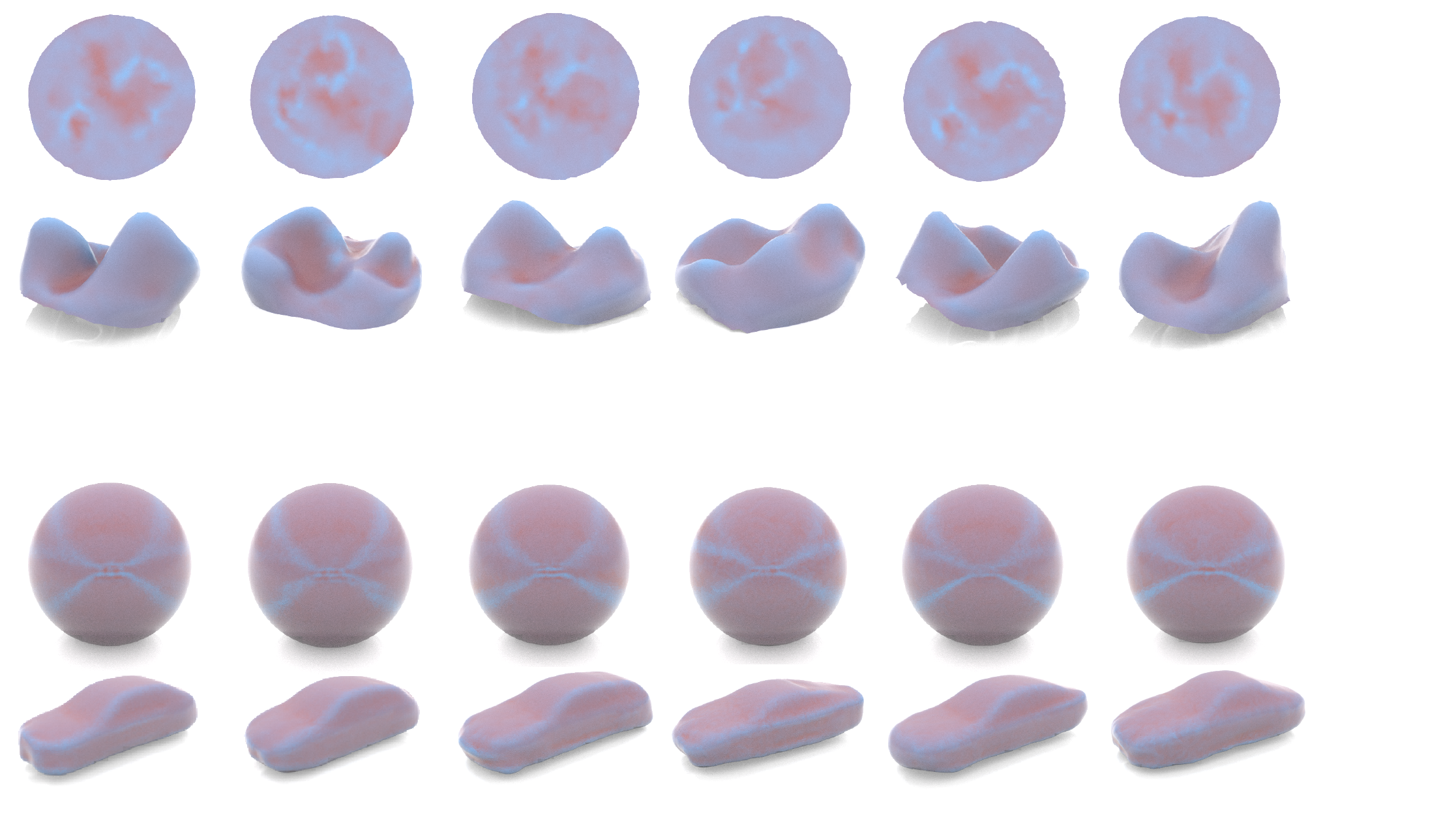}}%
    \put(0.88442842,0.3834496){\color[rgb]{0,0,0}\makebox(0,0)[lt]{\lineheight{1.25}\smash{\begin{tabular}[t]{l}\small mean curvature\end{tabular}}}}%
    \put(0.96832037,0.33004696){\color[rgb]{0,0,0}\makebox(0,0)[lt]{\lineheight{1.25}\smash{\begin{tabular}[t]{l}$1$\end{tabular}}}}%
    \put(0.96598704,0.18155287){\color[rgb]{0,0,0}\makebox(0,0)[lt]{\lineheight{1.25}\smash{\begin{tabular}[t]{l}$-1$\end{tabular}}}}%
    \put(0.90239415,0.35892082){\color[rgb]{0,0,0}\makebox(0,0)[lt]{\lineheight{1.25}\smash{\begin{tabular}[t]{l}\small half-density\end{tabular}}}}%
    \put(0,0){\includegraphics[width=\unitlength,page=2]{car_teeth.pdf}}%
  \end{picture}%
\endgroup%

\caption{Randomly generated teeth and cars via the variational autoencoder. The first and third rows show the isotropic meshings, which are induced from the generated density function, with the generated mean curvature half-density. The second and fourth rows show the resulting reconstruction. The architectures of neural networks are modified from the traditional autoencoders in Table \ref{fig:sph_network} and \ref{fig:disk_network} to variational autoencoder. }\label{fig:car}
\end{figure*}

\clearpage

\section{Architectures}
\label{app:architecture}

\begin{table}[h]
	\centering
	\scalebox{0.8}{%
	\begin{tabular}{c|c|c}
		\hline
		Encoder \\ \hline
		layers  & input & output\\
		Conv2D ($4\times 4$)  & $\scriptstyle 320\times 32 \times 32 \times 2$ & $\scriptstyle 320\times 32\times 32\times 4$\\
		BatchNormalization  & & \\
		LeakyReLu  & & \\
		Conv2D ($4\times 4$) & $\scriptstyle 320\times 32\times 32\times 4$ & $\scriptstyle 320\times 16\times 16\times 8$\\
		BatchNormalization  & & \\
		LeakyReLu  & & \\
		Conv2D ($4\times 4$) & $\scriptstyle 320\times 16\times 16\times 8$ & $\scriptstyle 320\times 8\times 8\times 16$\\
		BatchNormalization & & \\
		LeakyReLu  & & \\
		Conv2D ($4\times 4$) & $\scriptstyle 320\times 8\times 8\times 16$ & $\scriptstyle 320\times 4\times 4\times 32$\\
		BatchNormalization  & & \\
		LeakyReLu  & & \\
		Conv2D ($4\times 4$) & $\scriptstyle 320\times 4\times 4\times 32$ & $\scriptstyle 320\times 2\times 2\times 64$\\
		BatchNormalization  & & \\
		LeakyReLu  & & \\
		Conv2D ($4\times 4$) & $\scriptstyle 320\times 2\times 2\times 64$ & $\scriptstyle 320\times 1\times 1\times 128$\\
		BatchNormalization  & & \\
		LeakyReLu  & & \\
		Reshape  & $\scriptstyle 320 \times 1\times 1\times 128$ & $\scriptstyle 80\times 512$\\
		FC  & $\scriptstyle 80\times 512$ & $\scriptstyle 80\times 256$ \\
		BatchNormalization  & & \\
		LeakyReLu  & & \\
		Reshape  & $\scriptstyle 80\times 256$ & $\scriptstyle20\times 1024$\\
		FC  & $\scriptstyle 20\times 1024$ & $\scriptstyle20\times 512$ \\
		BatchNormalization  & & \\
		LeakyReLu  & & \\
		FC  & $\scriptstyle20\times 512$ & $\scriptstyle 200$ \\ 
		\vspace{1cm} & & \\ \hline
		Decoder \\ \hline
		layers  & input & output\\
		FC  & $\scriptstyle 200$ & $\scriptstyle 20480$ \\
		BatchNormalization & & \\
		LeakyReLu & & \\
		Reshape & $\scriptstyle 20480$ & $\scriptstyle 20\times 1024$ \\
		FC & $\scriptstyle 20 \times 1024$ & $\scriptstyle 20\times 2048$ \\
		BatchNormalization & & \\
		LeakyReLu & & \\
		Reshape & $\scriptstyle 20\times 2048$ & $\scriptstyle 80\times 512$ \\
		FC & $\scriptstyle 80 \times 512$ & $\scriptstyle 80 \times 1024$ \\
		BatchNormalization & & \\
		LeakyReLu & & \\
		Reshape & $\scriptstyle 80\times 1024$ & $\scriptstyle 320\times 2\times 2 \times 64$ \\
		Deconv2D ($4\times 4$)  & $\scriptstyle 320\times 2\times 2 \times 64$ & $\scriptstyle 320\times 4\times 4\times 32$ \\
		BatchNormalization & & \\
		LeakyReLu & & \\
		Deconv2D ($4\times 4$)  & $\scriptstyle 320\times 4\times 4 \times 32$ & $\scriptstyle 320\times 8\times 8\times 16$ \\
		BatchNormalization & & \\
		LeakyReLu & & \\
		Deconv2D ($4\times 4$)  & $\scriptstyle 320\times 8\times 8 \times 16$ & $\scriptstyle 320\times 16\times 16\times 8$ \\
		BatchNormalization & & \\
		LeakyReLu & & \\
		Deconv2D ($4\times 4$)  & $\scriptstyle 320\times 16\times 16 \times 8$ & $\scriptstyle 320\times 32\times 32\times 4$ \\
		BatchNormalization & & \\
		LeakyReLu & & \\
		Deconv2D ($4\times 4$)  & $\scriptstyle 320\times 32\times 32 \times 4$ & $\scriptstyle 320\times 32\times 32\times 2$
	\end{tabular}}
	\caption{The architecture for spherical surfaces.}\label{fig:sph_network}
\end{table}

\begin{table}[h]
	\centering
	\begin{tabular}{c|c|c}
		\hline
		Encoder \\ \hline
		layers  & input & output\\
		Conv2D ($4\times 4$)  & $\scriptstyle 256\times 256 \times 2$ & $\scriptstyle 128\times 128 \times 4$\\
		BatchNormalization  & & \\
		LeakyReLu  & & \\
		Conv2D ($4\times 4$)  & $\scriptstyle 128\times 128 \times 4$ & $\scriptstyle 64\times 64 \times 8$\\
		BatchNormalization  & & \\
		LeakyReLu  & & \\
		Conv2D ($4\times 4$)  & $\scriptstyle 64\times 64 \times 8$ & $\scriptstyle 32\times 32 \times 16$\\
		BatchNormalization  & & \\
		LeakyReLu  & & \\
		Conv2D ($4\times 4$)  & $\scriptstyle 32\times 32 \times 16$ & $\scriptstyle 16\times 16 \times 32$\\
		BatchNormalization  & & \\
		LeakyReLu  & & \\
		FC  & $\scriptstyle 16\times 16 \times 32$ & $\scriptstyle 100$
	\end{tabular}
	
	\bigskip
	
	\begin{tabular}{c|c|c}
		\hline
		Decoder \\ \hline
		layers  & input & output\\
		FC  & $\scriptstyle 100$ & $\scriptstyle 8192$ \\
		BatchNormalization & & \\
		LeakyReLu & & \\
		Reshape & $\scriptstyle 8192$ & $\scriptstyle 16\times 16 \times 32$ \\
		Deconv2D ($4\times 4$)  & $\scriptstyle 16\times 16\times 32$ & $\scriptstyle 32\times 132\times 16$ \\
		BatchNormalization & & \\
		LeakyReLu & & \\
		Deconv2D ($4\times 4$)  & $\scriptstyle 32\times 32\times 16$ & $\scriptstyle 64\times 64 \times 8$ \\
		BatchNormalization & & \\
		LeakyReLu & & \\
		Deconv2D ($4\times 4$)  & $\scriptstyle 64\times 64\times 8$ & $\scriptstyle 128\times 128\times 4$ \\
		BatchNormalization & & \\
		LeakyReLu & & \\
		Deconv2D ($4\times 4$)  & $\scriptstyle 128\times 128\times 4$ & $\scriptstyle 256\times 256\times 2$ \\
	\end{tabular}
	\caption{The architecture for disk-like surfaces.}\label{fig:disk_network}\end{table}

\begin{table}
	\centering
	\begin{tabular}{c|c|c}
		\hline
		Encoder \\ \hline
		layers  & input & output\\
		Conv3D ($4\times 4\times 4$)  & $\scriptstyle 193\times 80 \times 195 \times 1$ & $\scriptstyle 97\times 40 \times 98\times 4$\\
		BatchNormalization  & & \\
		LeakyReLu  & & \\
		Conv3D ($4\times 4\times 4$)  & $\scriptstyle 97\times 40 \times 98 \times 4$ & $\scriptstyle 49\times 20 \times 49\times 8$\\
		BatchNormalization  & & \\
		LeakyReLu  & & \\
		Conv3D ($4\times 4\times 4$)  & $\scriptstyle 49\times 20 \times 49 \times 8$ & $\scriptstyle 25\times 10 \times 25\times 16$\\
		BatchNormalization  & & \\
		LeakyReLu  & & \\
		Conv3D ($4\times 4\times 4$)  & $\scriptstyle 25\times 10 \times 25 \times 16$ & $\scriptstyle 13\times 5 \times 13\times 32$\\
		BatchNormalization  & & \\
		LeakyReLu  & & \\
		Conv3D ($4\times 4\times 4$)  & $\scriptstyle 213\times 5 \times 13 \times 32$ & $\scriptstyle 7\times 3  \times 7\times 64$\\
		BatchNormalization  & & \\
		LeakyReLu  & & \\
		FC  & $\scriptstyle 7\times 3 \times 7 \times 64$ & $\scriptstyle 200$
	\end{tabular}
	\caption{The architecture for volumetric data.}\label{fig:vol_network}
\end{table}

\newpage

\end{document}